\documentclass[12pt]{article}
\usepackage{graphicx,multicol}
\usepackage{epstopdf}
\usepackage{amsmath,amssymb,amsfonts,thumbpdf}
\usepackage{amsthm}
\usepackage{enumerate}
\usepackage{bm}
\usepackage{multirow}
\usepackage[authoryear]{natbib}
\usepackage[colorlinks,citecolor=blue]{hyperref}
\usepackage{color}
\usepackage{algorithm}
\usepackage{algorithmic}
\usepackage{float,subcaption}
\usepackage{url, xcolor}
\usepackage{rotating}

\usepackage{natbib}

\newtheorem{lemma}{Lemma}{\bf}{\it}
\newtheorem{theorem}{Theorem}{\bf}{\it}
{\bf}{\it}

\theoremstyle{definition}
{\bf}{\it}
{\bf}{\it}
{\bf}{\it}
\newtheorem{remark}{Remark}{\bf}{\it}

\def\Pr{\mathrm{Pr}}

\usepackage[left=1in, right=1in]{geometry}
\graphicspath{{sim_plot/}}
\allowdisplaybreaks

\begin{document}

\def\spacingset#1{\renewcommand{\baselinestretch}%
{#1}\small\normalsize} \spacingset{1}



  \title{\LARGE\bf STRAW: Structure-Adaptive Weighting Procedure for Large-Scale Spatial Multiple Testing}
  \author{Pengfei Wang, Pengyu Yan and Canhui Li\footnote{
  		Author for Correspondence: Canhui Li, E-mail: lich380@nenu.edu.cn. School of Mathematics and Statistics, Henan University, Kaifeng 475004, China.}
}

  \maketitle

\bigskip
\begin{abstract}
The problem of large-scale spatial multiple testing is often encountered in various scientific research fields, where the signals are usually enriched on some regions while sparse on others. To integrate spatial structure information from nearby locations, we propose a novel approach, called {\bf STR}ucture-{\bf A}daptive {\bf W}eighting (STRAW) procedure, for large-scale spatial multiple testing. The STRAW procedure is capable of handling a broad range of spatial settings by leveraging a class of weighted $p$-values and is fully data-driven. Theoretical results show that the proposed method controls the false discovery rate (FDR) at the pre-specified level under some mild conditions. In practice, the local sparsity level, defined as the probability of the null hypothesis being not true, is commonly unknown. To address this issue, we develop a new method for estimating the local sparsity level by employing the kernel-smooth local false discovery rate (Lfdr) statistic. The superior numerical performance of the STRAW procedure is demonstrated by performing extensive simulation studies and a real data analysis.
\end{abstract}

\noindent%
\vfill

\newpage
\spacingset{1.45} 
\section{Introduction} \label{sec:intro}
\par
The problem of large-scale spatial multiple testing commonly arises from a variety of scientific research fields. For example, in the field of climatology, climate scientists may need to tests tens of thousands of hypotheses to identify the locations where the mean and teleconnection are different between the reconstructed and targeted climate fields (Yun {\it et al.,} 2022). Some other application fields of large-scale spatial multiple testing include functional magnetic resonance imaging (fMRI) data analysis (Tansey {\it et al.,} 2018), genome-wide association studies (GWAS,  Sesia {\it et al.,} 2021), environmental studies (Deng {\it et al.,} 2022), among others. In such applications, tens of thousands of spatially correlated hypothesis tests are conducted simultaneously. It is interesting and challenging to exploit informative spatial patterns to improve both power and interpretability.

\par
To accommodate multiplicity, Benjamini \& Hochberg (1995) proposed a control criterion, called FDR, for multiple testing. The FDR is defined as the expected proportion of false rejections among all rejections and has become widely used in large-scale multiple testing. To date, a considerable amount of multiple testing procedures allowing control of the FDR have been proposed, including the Benjamini-Hochberg (BH) procedure (Benjamini \& Hochberg, 1995), the adaptive $p$-value procedure (Benjamini \& Hochberg, 2000), the $q$-value based procedure (Storey, 2002), the Lfdr procedure (Efron {\it et al.,} 2001), among others. It is important to note that these multiple testing procedures largely ignored spatial structure information among tests and implicitly assumed that the tests are independent. However, such an assumption is rarely established in practice. For example, since the neighbouring SNPs tend to segregate into the same gametes during meiosis, the observations arising from GWAS often exhibit complex serial correlations (Wei {\it et al.,} 2009). As a result, the direct use of the traditional methods for spatial multiple testing may lead to loss of testing efficacy and distortion of scientific findings. On the other hand, the proper use of spatial patterns is expected to improve the accuracy of the tests and the interpretability of the findings.

\par
To date, incorporating spatial structure information into multiple testing has become an active research area. One way of describing spatial structure among tests is to use probabilistic graphical models (PGMs). For example,  Sun \& Cai (2009) assumed that the state sequence of the null hypotheses follows a first-order Markov chain and developed a method based on hidden Markov models (HMMs) for one-dimensional spatial multiple testing. To allow for two-dimensional spatial correlations,  Sun {\it et al.} (2015) proposed a multiple testing procedure based on random fields. For some other multiple testing procedures based on PGMs, see Kuan \& Chiang (2012); Wang {\it et al.} (2019); Wang \& Zhu (2019); Cui {\it et al.} (2021). However, it is necessary to note that these PGM-based multiple testing procedures are sensitive to model misspecification (Yun {\it et al.,} 2022).

\par
Essentially, the idea behind spatial multiple testing procedures is to use spatial structure information to relax the $p$-value threshold for hypotheses that are more likely to be alternatives and tighten the threshold for other hypotheses (Deng {\it et al.,} 2022). In this regard, Tansey {\it et al.} (2018) proposed a FDR smoothing approach that automatically detects spatially localized regions of significant test statistics and subsequently relaxes the significance thresholds within these regions. However, it is worth noting that the FDR smoothing approach requires estimations of both the null and non-null densities, and a poor estimate of these densities can result in less powerful and even invalid FDR procedures. Recently, a number of weighted $p$-values methods that employ auxiliary/side information for multiple testing have been proposed (Ignatiadis {\it et al.,} 2016; Li \& Barber, 2019; Xia {\it et al.,} 2020). To exploit spatial structure information in multiple testing, Cai {\it et al.} (2022) proposed an approach that weights $p$-values based on the local sparsity level. They showed that the proposed approach, called locally adaptive weighting and screening rule (LAWS) hereafter, can adaptively learn the sparse structure without requiring estimates of the null and non-null densities. Similar idea has also been extended to incorporate side information in recent studies, such as the work by Cao {\it et al.} (2022); Zhang \& Chen (2022).

\par
In this paper, we propose an approach, called {\bf STR}ucture-{\bf A}daptive {\bf W}eighting (STRAW) procedure, along the line of the research on weighted $p$-values. The STRAW procedure provides a novel perspective of exploiting spatial structure information by a class of weighted $p$-values and is fully data-driven. The class of the proposed weighted $p$-values contains two parts of information, that is, the
information on the signal strength and the information on the spatial structure. Overall, this paper offers three main contributions. Firstly, a class of general weighted $p$-values which incorporates spatial structure information has been proposed and the weighted $p$-value employed in Cai {\it et al.} (2022) is a special case of our proposal. Secondly, we develop a general method for estimating the local sparsity level by leveraging the kernel-smooth Lfdr statistic. In practice, the kernel functions can be flexibly selected according to the assumption of the spatial structure among tests. Thirdly, we provide theoretical evidence that the STRAW procedure is capable of controlling the FDR at a pre-specified significance level and achieves higher ranking efficacy compared the BH procedure. We further perform extensive simulation studies and a real data analysis to demonstrate the superior numerical performance of the STRAW procedure.

\par
The rest of this paper is organized as follows. Section 2 begins by giving a description of the problem formulation for large-scale spatial multiple testing, followed by a brief introduction to the Lfdr procedure and its equivalent decision rule. In the subsequent part of Section 2, we introduce in detail the oracle STRAW procedure (the local sparsity level is known in advance) and its theoretical properties, the estimation of the local sparsity levels, and the data-driven STRAW procedure (the local sparsity level is estimated by using the kernel-smooth Lfdr statistic). A series of simulation studies and a real data analysis are presented in Sections 3 and 4, respectively. Finally, we summarize this work about the STRAW procedure in Section 5.

\newpage

\section{Statistical Methods}\label{sec:model}

\subsection{Problem Formulation}\label{sec-2.1}

\par
In this section, we shall introduce the problem formulation for large-scale spatial multiple testing. Consider a region in $d$-dimensional Euclidean space, that is, $\mathcal{S}\subset R^d$. Let $\mathbb{S}$ be a regular lattice in the region $\mathcal{S}$ and assume that $\mathbb{S}\rightarrow\mathcal{S}$ in the infill-asymptotics framework (Stein, 1999). Consider simultaneous testing of finite null and non-null hypotheses, namely, $\{H_0(s)\}_{s\in\mathbb{S}}$ and $\{H_1(s)\}_{s\in\mathbb{S}}$, whose locations are located in regular lattice $\mathbb{S}$. Let $\{\theta(s)\}_{s\in\mathbb{S}}$ be a list of underlying states of the null hypothesis, where $\theta(s)=0$ implies that the null hypothesis $H_0(s)$ is true and $\theta(s)=1$ otherwise. Let $\{p(s)\}_{s\in\mathbb{S}}$ be a list of $p$-values with respect to the test located at $s$. The problem of spatial multiple testing can be formulated as:
\[
H_0(s): \theta(s)=0\text{~~v.s.~~} H_1(s): \theta(s)=1, \text{~for~} s\in\mathbb{S}.
\]
In fact, a multiple testing procedure can be expressed as a decision rule $\{\delta(s)\}_{s\in\mathbb{S}}$, where $\delta(s)=0$ implies that the null hypothesis $H_0(s)$ is not rejected by the decision rule $\{\delta(s)\}_{s\in\mathbb{S}}$ and $\delta(s)=1$ otherwise. Then the FDR and the marginal false discovery rate (mFDR) of the decision rule $\{\delta(s)\}_{s\in\mathbb{S}}$ are respectively defined as
\[
\mathrm{FDR} = \mathrm{E} \left[\frac{\sum_{s\in\mathbb{S}}(1-\theta(s))\delta(s)}{\sum_{s\in\mathbb{S}}\delta(s)\vee1}\right] ~\text{and}~ \mathrm{mFDR} = \frac{\mathrm{E}\left[\sum_{s\in\mathbb{S}}(1-\theta(s))\delta(s)\right]}{\mathrm{E}\left[\sum_{s\in\mathbb{S}}\delta(s)\right]}.
\]
In general, the purpose of large-scale spatial multiple testing is to develop a method that is capable of controlling the FDR or the mFDR at the pre-specified level and achieves higher power by leveraging spatial structure information among tests. Due to the lack of consideration of spatial structure information, the direct use of traditional multiple testing procedures may lead to a decrease in testing accuracy and interpretability of findings.

\subsection{Lfdr Procedure and Its Equivalent Decision Rule}\label{sec-2.2}

\par
In this section, we will briefly describe the Lfdr procedure and its equivalent decision rule to clarify the motivation for the approach proposed in this paper. Consider the two-component mixture model proposed by Efron {\it et al.} (2001):
\begin{eqnarray*}
	\theta(s)        &\overset{iid}{\sim}& \text{Bernoulli}(\pi_1), \text{~for~} s\in\mathbb{S},\\
	p(s)\mid\theta(s) &\sim& (1-\theta(s)) f_0 + \theta(s) f_1, \text{~for~} s\in\mathbb{S},
\end{eqnarray*}
where $\pi_1=\Pr(\theta(s)=1)$, for all $s\in\mathbb{S}$, and $f_0$ and $f_1$ are the probability density functions (PDFs) for the null and non-null hypotheses, respectively. The local false discovery rate (Lfdr) statistic is defined as:
\[
\text{Lfdr}(s) = \Pr (\theta(s) = 0 \mid p(s)) = \frac{(1-\pi_1)f_0(p(s))}{(1-\pi_1)f_0(p(s))+\pi_1 f_1(p(s))}, \text{~for~} s\in\mathbb{S}.
\]
Then the Lfdr procedure is executed as follows:
\[
\text{let}~l=\max\left\{j:\frac{1}{j}\sum\limits^j_{i=1}\text{Lfdr}_{(i)}\leq\alpha\right\};~\text{then~reject~all~} H_0{(i)}, \text{~for~} i=1,\cdots,l,\eqno{(1)}
\]
where $\text{Lfdr}_{(i)}$ is the $i$th smallest Lfdr statistic and $H_{0(i)}$ is the corresponding null hypothesis. By a simple derivation, we can obtain a decision rule equivalent to the Lfdr procedure, namely,
\[
\text{~for~} s\in\mathbb{S}, \text{~if~}\frac{f_1(p(s))}{f_0(p(s))} \geq \frac{1-t_{\alpha}}{t_{\alpha}} \cdot \frac{1-\pi_1}{\pi_1}, \text{~then~reject~} H_0(s), \eqno{(2)}
\]
where $t_{\alpha}$ is a cutoff which satisfies $\text{Lfdr}_{(l)}\leq t_{\alpha}<\text{Lfdr}_{(l+1)}$, and $l$ is specified by the Lfdr procedure (1). When the tests are independent, Sun \& Cai (2007) have shown that the Lfdr procedure is optimal in the sense that it controls the mFDR at the pre-specified level with the smallest marginal false non-discovery rate (mFNR).

\subsection{Generalized Weighting Strategy}\label{sec-2.3}

\par
In large-scale spatial multiple testing, since the tests usually exhibit complex spatial correlations, the direct use of the Lfdr procedure may result in loss of testing efficacy. To allow for spatial structure information, the basic idea is to relax the $p$-value threshold for hypotheses where the surrounding signal (i.e. the null hypothesis is not true) is enriched and tighten the threshold for hypotheses where the surrounding signal is sparse.

\par
Consider the decision rule (2) described in Section 2.2. Note that the Lfdr procedure is uniquely determined by the likelihood ratio $f_1(\cdot)/f_0(\cdot)$ and the threshold function $(1-t_{\alpha})/t_{\alpha} \cdot (1-\pi_1)/\pi_1$. In order to enable the Lfdr procedure to accommodate spatial correlations, it is reasonable to assume that the likelihood ratio and the threshold function are location-adaptive, that is, they can be expressed as $f_{1,s}(\cdot)/f_{0,s}(\cdot)$ and ${(1-t^*_{\alpha})}/{t^*_{\alpha}} \cdot {(1-\pi_1(s))}/{\pi_1(s)}$, respectively, where ${f_{1,s}(\cdot)}$ and ${f_{0, s}(\cdot)}$ are location-specific PDFs for the null and the non-null, respectively, $t^*_{\alpha}$ is the cutoff with the FDR controlled at $\alpha$, and $\pi_1(s)=\Pr(\theta(s)=1)$ is the local sparsity level. However, it is extremely difficult to model and estimate the location-adaptive likelihood ratio $f_{1,s}(\cdot)/f_{0,s}(\cdot)$ directly. To bypass this issue, Cai {\it et al.} (2022) suggested to replace $f_{1,s}(p)/f_{0,s}(p)$, for $s\in\mathbb{S}$, by the surrogate function $p^{-1}$. In this paper, we suggest to use a more flexible surrogate function, that is, $p^{-k}$, where $k>0$ is a turning parameter. In the decision rule (2), substituting $p(s)^{-k}$ and ${(1-t^*_{\alpha})}/{t^*_{\alpha}} \cdot {(1-\pi_1(s))}/{\pi_1(s)}$ for $f_1(p(s))/f_0(p(s))$ and $(1-t_{\alpha})/t_{\alpha} \cdot (1-\pi_1)/\pi_1$, respectively. Then we can obtain the location-adaptive decision rule:
\[
\text{~for~} s\in\mathbb{S}, \text{~if~}  \left(\frac{1-\pi_1(s)}{\pi_1(s)}\right)^{\frac{1}{k}} p(s) \leq \left(\frac{t^*_{\alpha}}{1-t^*_{\alpha}}\right)^{\frac{1}{k}}, \text{~then~reject~} H_0(s). \eqno{(3)}
\]
In light of the above considerations, we define a class of weighted $p$-values as:
\[
p_{\text{weighted}}(s, k) = \min\left\{\left(\frac{1-\pi_1(s)}{\pi_1(s)}\right)^{\frac{1}{k}} p(s), 1\right\}, \text{~for~} s\in\mathbb{S}.
\]
\begin{remark}
	This class of weighted $p$-values is simple but thoughtful. Here, we make some notes for it.
	\begin{enumerate}[~~~~~~(a)]
		\item This class of weighted $p$-values provide a unified framework for incorporating spatial correlations. In particular, the weighted $p$-value proposed by Cai {\it et al.} (2022) is a special case with $k=1$, that is, $p_{\text{weighted}}(s, 1)$.
		\item All such weighted $p$-values contain two parts of information, one of which is the information on the strength of the signal from the $p$-value $p(s)$ and the other is the sparsity structure information from the term $\left(\frac{1-\pi_1(s)}{\pi_1(s)}\right)^{1/k}$.
		\item If there exists a small constant $\xi>0$ such that $\pi_1(s)\in[\xi, 1-\xi]$, for all $s\in \mathbb{S}$, then we have that $\lim\limits_{k\rightarrow+\infty} \left(\frac{1-\pi_1(s)}{\pi_1(s)}\right)^{1/k} = 1$. This suggests that when the value of $k$ is large, $p_{\text{weighted}}(s, k)$ will be almost free of spatial structure information. In other words, $p_{\text{weighted}}(s, k)$ is almost the same as the original $p$-value in such a case.
		\item The selection of the turning parameter $k$ is fully data-driven and we will illustrate how to choose it in Section 2.5.
	\end{enumerate}
\end{remark}

\subsection{Threshold Selection}\label{sec-2.4}

\par
In Sections 2.4 and 2.5, it is assumed that the local sparsity levels $\{\pi_1(s)\}_{s\in \mathbb{S}}$ are known and we will describe how to specify $t^*_{\alpha}$ so that the FDR of the location-adaptive decision rule (3) is controlled at $\alpha$. The estimation of $\{\pi_1(s)\}_{s\in \mathbb{S}}$ is deferred to Section 2.6.

\par
Let $\varphi_k(x)=[x/(1-x)]^{1/k}$. Since $\varphi_k(t^*_{\alpha})=[t^*_{\alpha}/(1-t^*_{\alpha})]^{1/k}$ increases strictly monotonically with respect to $t^*_{\alpha}$, it is sufficient to specify $\varphi_k(t^*_{\alpha})$ so that the corresponding FDR is controlled at $\alpha$. The location-adaptive decision rule (3) consists of two steps, where the first step is to order the weighted $p$-values
from the smallest to the largest, and the second step is to select a cutoff so that the FDR is controlled at $\alpha$. Let $p^{(j)}_{\text{weighted}}(k)$ be the $j$th smallest weighted $p$-value among $\{p_{\text{weighted}}(s, k)\}_{s\in\mathbb{S}}$, $p^{(0)}_{\text{weighted}}(k)=0$, and $m=|\mathbb{S}|$. Suppose that $\varphi_k(t^*_{\alpha})$ falls into some interval $[p^{(j)}_{\text{weighted}}(k), p^{(j+1)}_{\text{weighted}}(k))$, for $j=0,\cdots, m-1$. Then the false discovery proportion (FDP) of the location-adaptive decision rule (3) can be expressed as:
\begin{eqnarray*}
	\mathrm{FDP}  &=& \frac{\sum_{s\in\mathbb{S}}(1-\theta(s))I{(p_{\text{weighted}}(s, k)\leq \varphi_k(t^*_{\alpha}))}}{\max\left\{\sum_{s\in\mathbb{S}}I{(p_{\text{weighted}}(s, k)\leq \varphi_k(t^*_{\alpha}))}, 1\right\}} \\
	&=& \frac{\sum_{s\in\mathbb{S}}(1-\theta(s))I{(p_{\text{weighted}}(s, k)\leq \varphi_k(t^*_{\alpha}))}}{\max\left\{j, 1\right\}}.
\end{eqnarray*}
Since the underlying states of the null hypothesis (i.e. $\{\theta(s)\}_{s\in\mathbb{S}}$) are unobserved, the number of false positives (i.e. the numerator $\sum_{s\in\mathbb{S}}(1-\theta(s))I{(p_{\text{weighted}}(s, k)\leq \varphi_k(t^*_{\alpha}))}$) is unknown. We propose to replace it by its expectation, that is, the expected number of false positives (EFP). Specifically, the EFP can be expressed as
\begin{eqnarray*}
	\mathrm{EFP} &=& \mathrm{E} \left[ \sum_{s\in\mathbb{S}}(1-\theta(s))I{(p_{\text{weighted}}(s, k)\leq \varphi_k(t^*_{\alpha}))} \right]\\
	&=& \sum_{s\in\mathbb{S}}\Pr(\theta(s)=0)\Pr(p_{\text{weighted}}(s, k)\leq \varphi_k(t^*_{\alpha}) \mid \theta(s)=0)\\
	&=& \sum_{s\in\mathbb{S}} (1-\pi_1(s))\Pr(p(s)\leq \varphi_k(\pi_1(s))\varphi_k(t^*_{\alpha})\mid \theta(s)=0)\\
	&\leq& \sum_{s\in\mathbb{S}} (1-\pi_1(s)) \varphi_k(\pi_1(s))\varphi_k(t^*_{\alpha}),
\end{eqnarray*}
where the last inequality is due to the fact that $\varphi_k(\pi_1(s))\varphi_k(t^*_{\alpha})$ may be greater than $1$. Then the FDP can be estimated by $j^{-1} \sum_{s\in\mathbb{S}} (1-\pi_1(s)) \varphi_k(\pi_1(s))\varphi_k(t^*_{\alpha})$ and $\varphi_k(t^*_{\alpha})$ is specified by maximizing the number of discoveries within the constraints of the FDP estimates. Specifically,
{\small
	\[
	\text{let~} l_{k}=\max\left\{j: \frac{1}{j}\sum_{s\in\mathbb{S}} (1-\pi_1(s)) \varphi_k(\pi_1(s))p^{(j)}_{\text{weighted}}(k)\leq\alpha\right\}, \text{then specify~} \varphi_k(t^*_{\alpha}) = p^{(l_{k})}_{\text{weighted}}(k).\eqno{(4)}
	\]
}
Similar considerations for specifying the threshold could also be found in Liang {\it et al.} (2022).

\subsection{Oracle STRAW Procedure}\label{sec-2.5}

\par
Note that a specific turning parameter $k$ corresponds to a list of weighted $p$-values, that is, $\{p_{\text{weighted}}(s, k)
\}_{s\in\mathbb{S}}$. The selection of $k=1$ is equivalent to the weighted $p$-value proposed by Cai {\it et al.} (2022), whereas choosing a large value of $k$ amounts to using the non-weighted $p$-value. We suggest to select $k$ with the largest number of discoveries, namely,
\[
\widetilde{k} = \underset{k\in(0, +\infty)}{\arg\max}~l_{k}.\eqno{(5)}
\]
Since the values of $k$ are taken continuously on the interval $(0, +\infty)$, it is extremely difficult to calculate the exact value of $\widetilde{k}$. To address this issue, we adopt the approximation strategy. Note that when the value of $k$ is large, $p_{\text{weighted}}(s, k)$ will be almost free of spatial structure information. On the other hand, when the value of $k$ is small, most of $p_{\text{weighted}}(s, k)$ will tend to $1$. It is sufficient to consider the values of $k$ that are equally spaced on the interval $[B_1, B_2]$, namely, $\{k_i\}^L_{i=0}$, where $B_1>0$, and $B_2$ is a relatively large constant and $k_i=B_1+i(B_2-B_1)/L$, for $i=0,\cdots,L$. Specifically, the turning parameter $k$ is specified as follows.
\begin{enumerate}[(a)]
	\item Let $\{k_i\}^L_{i=0}$ be equally spaced on the interval $[B_1, B_2]$, where $k_i=B_1+i(B_2-B_1)/L$.
	\item For each $k_i$, rank $\{p_{\text{weighted}}(s, k_i)\}_{s\in\mathbb{S}}$ from the smallest to largest, that is,  $p_{\text{weighted}}^{(1)}(k_i), \cdots, \\p_{\text{weighted}}^{(m)}(k_i)$, where $m=|\mathbb{S}|$.
	\item Select $ \widetilde{k} = \underset{k_i, i=0,\cdots, L}{\arg\max}\left\{j: \dfrac{1}{j}\sum_{s\in\mathbb{S}} (1-\pi_1(s)) \varphi_{k_i}(\pi_1(s))p^{(j)}_{\text{weighted}}(k_i)\leq\alpha\right\}$.
\end{enumerate}

\begin{remark}
	It is important to note that a large value of $L$ chosen will result in a large computational burden, while a small choice of $L$ will result in a large deviation of the estimated $\widetilde{k}$ from the true $\widetilde{k}$. In general, we can choose the appropriate $L$ according to the complexity of the simulation and the required accuracy.
\end{remark}

\par
Integrate the above concerns (3)-(5), we can obtain the oracle {\bf STR}ucture-{\bf A}daptive {\bf W}eighting (STRAW) procedure, which is summarized in Algorithm 1.
\begin{algorithm}[htp]
	\begin{enumerate}[1.]\setlength\itemsep{0.2em}
		\item Rank $\{p_{\text{weighted}}(s, \widetilde{k})\}_{s\in\mathbb{S}}$ from smallest to largest.
		\item Let $l_{\widetilde{k}}=\max\left\{j: \dfrac{1}{j}\sum_{s\in\mathbb{S}} (1-\pi_1(s)) \varphi_{\widetilde{k}}(\pi_1(s))p^{(j)}_{\text{weighted}}(\widetilde{k})\leq\alpha\right\}$.
		\item For $s\in\mathbb{S}$, reject the null hypotheses for which $p_{\text{weighted}}(s, \widetilde{k})\leq p^{(l_{\widetilde{k}})}_{\text{weighted}}(\widetilde{k})$.
	\end{enumerate}
	\caption{The oracle STRAW procedure}
	\label{alg:1}
\end{algorithm}

\par
Let $\{s_1, \cdots, s_m\}$ be an  arrangement of $\{s\in\mathbb{S}\}$ and $z_i=\Phi^{-1}(1-p(s_i)/2)$, where $m=|\mathbb{S}|$. Let $\mathbb{S}_{\rho}$ be the set of indices with $|\mathrm{E}(z_i)|\geq (\log m)^{(1+\rho)/2}$, that is,
\[
\mathbb{S}_{\rho} = \left\{i: 1\leq i\leq m, |\mathrm{E}(z_i)|\geq (\log m)^{(1+\rho)/2}\right\}.
\]
Let $F_{1,i}(\cdot)$ be the location-adaptive cumulative distribution function (CDF) under the non-null with respect to $p(s_i)$, that is, $F_{1,i}(t)=\Pr(p(s_i)\leq t\mid \theta(s_i)=1)$. Next, we discuss some assumptions that are needed in deriving the properties of the oracle STRAW procedure.

\begin{enumerate}[(\text{A}1)]
	\item Assume that
	\[
	\sum\limits_{s\in\mathbb{S}}\Pr(\theta(s)=0)\geq cm, \text{~and~} \mathrm{Var}\left(\sum\limits_{s\in\mathbb{S}}I(\theta(s)=0)\right)=O(m^{1+\zeta}),
	\]
	for some constant $c>0$ and some $0\leq\zeta<1$.
	\item Assume that
	\[
	|\mathbb{S}_{\rho}| \geq \left[1/(\pi^{1/2}\alpha)+\delta\right](\log m)^{1/2},
	\]
	for some constant $\rho>0$ and some $\delta>0$, where $\pi$ is the circular constant.
	\item Assume that
	\[
	\sum^m_{i=1}a_i F_{1,i}(t/x_i)\geq F_{1,i}\left( t/\sum^m_{i=1} a_i x_i\right)
	\]
	for any $0< a_i< 1$ such that $\sum^m_{i=1} a_i =1$, and $\min_{1\leq j \leq m} 1/\varphi_{\widetilde{k}}(\pi_1(s_j)) \leq x_i \leq \max_{1\leq j \leq m} 1/\varphi_{\widetilde{k}}(\pi_1(s_j))$.
	\item Assume that
	\[
	\sum_{s\in\mathbb{S}}(1-\pi_{1}(s))\sum_{s\in\mathbb{S}}\pi_{1}(s) \geq \sum_{s\in\mathbb{S}}\left[(1-\pi_{1}(s))^{1-1/\widetilde{k}}\pi_{1}(s)^{1/\widetilde{k}}\right]
	\sum_{s\in\mathbb{S}}\left[(1-\pi_{1}(s))^{1/\widetilde{k}}\pi_{1}(s)^{1-1/\widetilde{k}}\right].
	\]
\end{enumerate}

\begin{remark}
	Assumption (A1) is fairly mild, as it only requires that the underlying states $\{\theta(s)\}_{s\in\mathbb{S}}$ are not completely correlated. Assumption (A2) assume that there exist a small number of locations with $\mathrm{E}(z_i)|\geq (\log m)^{(1+\rho)/2}$ for some $\rho>0$. When the CDFs under the non-null are homogeneous, namely, $F_{1,i}(t)=F_{1}(t)$, Assumption (A3) reduces to the assumption that $x\mapsto F_{1}(t/x)$ is a convex function. Such an assumption is satisfied by the CDFs employed in Hu {\it et al.} (2010). If $\widetilde{k}=2$, then one can prove that Assumption (A4) is valid by using Cauchy-Schwartz inequality. For the general case, Assumption (A4) can be verified by using the estimation of $\{\pi_1(s)\}_{s\in\mathbb{S}}$.
\end{remark}

\par
Let $\boldsymbol{\delta}^{\widetilde{k}} = \left\{I(p_{\text{weighted}}(s, \widetilde{k})\leq p^{(l_{\widetilde{k}})}_{\text{weighted}}(\widetilde{k})), s\in\mathbb{S}\right\}$ be the decision rule corresponding to the oracle STRAW procedure. The next theorem shows that $\boldsymbol{\delta}^{\widetilde{k}}$ is capable of controlling both the FDR and FDP at level $\alpha$ under some mild conditions.
\begin{theorem}
	Consider the oracle STRAW procedure and assume that Assumptions (A1) and (A2) hold, then we have
	\[
	\underset{m\rightarrow\infty}{\lim\sup}~\mathrm{FDR}(\boldsymbol{\delta}^{\widetilde{k}}) \leq \alpha, \text{~and~} \lim\limits_{m\rightarrow\infty} \Pr (\mathrm{FDP}(\boldsymbol{\delta}^{\widetilde{k}}) \leq \alpha+\varepsilon) = 1,
	\]
	for any $\varepsilon>0$.
\end{theorem}

\par
Let $\omega_{\widetilde{k}} = \left[\frac{\sum_{s\in\mathbb{S}}(1-\pi_1(s))}{\sum_{s\in\mathbb{S}}\varphi_{\widetilde{k}}(\pi_1(s))(1-\pi_1(s))}\right]$ and $p_{\text{rescaled}}(s, \widetilde{k}) = p_{\text{weighted}}(s, \widetilde{k})/\omega_{\widetilde{k}}$ be the rescaled weighted $p$-values. It is clear that $\{p_{\text{rescaled}}(s, \widetilde{k})\}_{s\in\mathbb{S}}$ and $\{p_{\text{weighted}}(s, \widetilde{k})\}_{s\in\mathbb{S}}$ share the same ranking order. Let $\boldsymbol{\delta}_{\text{rescaled}}^{\widetilde{k}}(t) = \left\{I(p_{\text{rescaled}}(s, \widetilde{k})\leq t), s\in\mathbb{S}\right\}$ and $\boldsymbol{\delta}(t) = \left\{I(p(s)\leq t), s\in\mathbb{S} \right\}$ be the decision rules based on $\{p_{\text{rescaled}}(s, \widetilde{k})\}_{s\in\mathbb{S}}$ and $\{p(s)\}_{s\in\mathbb{S}}$, respectively. Then the mFDR of $\boldsymbol{\delta}_{\text{rescaled}}^{\widetilde{k}}(t)$ and $\boldsymbol{\delta}(t)$ can be expressed as:
\[
\mathrm{mFDR} (\boldsymbol{\delta}_{\text{rescaled}}^{\widetilde{k}}(t)) = \dfrac{\sum_{i=1}^m (1-\pi_1(s_i))t}{\sum_{i=1}^m (1-\pi_1(s_i))t + \sum_{i=1}^m \pi_1(s_i) F_{1, i}(\varphi_{\widetilde{k}}(\pi_1(s_i))\omega_{\widetilde{k}}t)},
\]
and
\[
\mathrm{mFDR} (\boldsymbol{\delta}(t)) = \dfrac{\sum_{i=1}^m (1-\pi_1(s_i)) t}{\sum_{i=1}^m (1-\pi_1(s_i)) t + \sum_{i=1}^m \pi_1(s_i) F_{1, i} (t)},
\]
respectively. Define the oracle thresholds $t^{\widetilde{k}}_{\text{OR}} = \sup\{t: \mathrm{mFDR} (\boldsymbol{\delta}_{\text{rescaled}}^{\widetilde{k}}(t)) \leq \alpha\}$ and $t^{1}_{\text{OR}} = \sup\{t: \mathrm{mFDR} (\boldsymbol{\delta}(t)) \leq \alpha\}$, respectively. The expected number of true positives (ETP) of $\boldsymbol{\delta}_{\text{rescaled}}^{\widetilde{k}}(t)$ and $\boldsymbol{\delta}(t)$ are defined as:
\[
\mathrm{ETP} (\boldsymbol{\delta}_{\text{rescaled}}^{\widetilde{k}}(t)) = \sum_{i=1}^m \pi_1(s_i) F_{1, i}(\varphi_{\widetilde{k}}(\pi_1(s_i))\omega_{\widetilde{k}}t),
\]
and
\[
\mathrm{ETP} (\boldsymbol{\delta}(t)) = \sum_{i=1}^m \pi_1(s_i) F_{1, i} (t),
\]
respectively. The next theorem shows that the weighting strategy provided by the rescaled weighted $p$-values achieves higher ranking efficiency compared with the unweighted $p$-values.

\begin{theorem}
	Consider the oracle STRAW procedure and assume that Assumptions (A3) and (A4) hold, then we have
	\begin{eqnarray*}
		&\mathrm{(1)}& \mathrm{mFDR} (\boldsymbol{\delta}_{\text{rescaled}}^{\widetilde{k}}(t^{1}_{\text{OR}})) \leq \mathrm{mFDR} (\boldsymbol{\delta}(t^{1}_{\text{OR}})) \leq \alpha,\\
		&\mathrm{(2)}& \mathrm{ETP} (\boldsymbol{\delta}_{\text{rescaled}}^{\widetilde{k}}(t^{\widetilde{k}}_{\text{OR}})) \geq \mathrm{ETP} (\boldsymbol{\delta}_{\text{rescaled}}(t^{1}_{\text{OR}})) \geq \mathrm{ETP} (\boldsymbol{\delta}(t^{1}_{\text{OR}})).
	\end{eqnarray*}
\end{theorem}

\par
The detailed proofs of Theorems 1 and 2 are given in Appendix.

\subsection{Estimation of Local Sparsity Levels}

\par
In practice, the local sparsity levels $\{\pi_1(s)\}_{s\in \mathbb{S}}$ are typically unknown. Since we have only one observation at each location, it is not an easy task to estimate $\{\pi_1(s)\}_{s\in \mathbb{S}}$ directly. In this section, we introduce a nonparametric approach to estimate the local sparsity levels.

\par
To leverage spatial information from nearby locations, we employ the so-called Nadaraya-Watson kernel-weighted average to assign weights to $1-\widehat{\text{Lfdr}}(s^{'})$ on the basis of their distance to location $s$, where $\widehat{\text{Lfdr}}(s^{'})$ is the estimated Lfdr statistic at location $s^{'}$. Specifically, the local sparsity level $\pi_1(s)$ is estimated by
\[
\widehat{\pi}_1(s) = \dfrac{\sum_{s^{'}\in\mathbb{S}}K_h(s-s^{'})I(|s-s^{'}|<c)(1-\widehat{\text{Lfdr}}(s^{'}))}{\sum_{s^{'}\in\mathbb{S}}K_h(s-s^{'})I(|s-s^{'}|<c)},
\]
where $K_h(\cdot)$ is the Gaussian kernel function, $h$ is the bandwidth, and $c$ is the truncation. Note that $1-\text{Lfdr}(s^{'})=\Pr(\theta(s^{'})=1\mid p(s^{'}))$, it is clear that the estimate of $\pi_1(s)$ can be viewed as an weighted average of $1-\widehat{\text{Lfdr}}(s^{'})$, where the weight decreases with the distance from $s^{'}$ to $s$. Such a similar consideration could also be found in (G$\ddot{\text{o}}$lz {\it et al.,} 2022).

\par
It is necessary to note that the basic idea for estimating the local sparsity levels is to assign weights that decay smoothly with the distance from the target location. In fact, the choice of the kernel function is not limited to Gaussian kernels. For example, one can choose the Epanechnikov quadratic kernel function which is defined as
$$ K_{\lambda}(s, s^{'})=
\begin{cases}
	\dfrac{3}{4}\left(1-\dfrac{|s-s^{'}|^2}{\lambda^2}\right),~& \text{if~} |s-s^{'}|/\lambda\leq 1,\\
	0,~& \text{otherwise.~}
\end{cases}$$
In the following simulation studies, we will fully verify the feasibility of choosing a Gaussian kernel. Although the choice of different kernel functions may affect the accuracy of the estimation, this aspect is beyond the scope of this paper.

\subsection{Data-Driven STRAW Procedure}

\par
Based on the estimate of the local sparsity levels, the estimated weighted $p$-values can be calculated by
\[
\widehat{p}_{\text{weighted}}(s, k) = \min\left\{\left(\frac{1-\widehat{\pi}_1(s)}{\widehat{\pi}_1(s)}\right)^{\frac{1}{k}} p(s), 1\right\}, \text{~for~} s\in\mathbb{S}.
\]
Then the turning parameter $k$ is estimated by
\[
\widehat{k} = \underset{k_i, i=0,\cdots, L}{\arg\max}\left\{j: \frac{1}{j}\sum_{s\in\mathbb{S}} (1-\widehat{\pi}_1(s)) \varphi_{k_i}(\widehat{\pi}_1(s))\widehat{p}^{(j)}_{\text{weighted}}(k_i)\leq\alpha\right\},
\]
where $\widehat{p}^{(j)}_{\text{weighted}}(k)$ is the $j$th smallest estimated weighted $p$-value among $\{\widehat{p}_{\text{weighted}}(s, k)\}_{s\in\mathbb{S}}$, and $k_i=B_1+i(B_2-B_1)/L$. The data-driven STRAW procedure proceeds as follows.
\begin{algorithm}[h]
	\begin{enumerate}[1.]\setlength\itemsep{0.2em}
		\item Rank $\{\widehat{p}_{\text{weighted}}(s, \widehat{k})\}_{s\in\mathbb{S}}$ from smallest to largest.
		\item Let $l_{\widehat{k}}=\max\left\{j: \dfrac{1}{j}\sum_{s\in\mathbb{S}} (1-\widehat{\pi}_1(s)) \varphi_{\widehat{k}}(\widehat{\pi}_1(s))\widehat{p}^{(j)}_{\text{weighted}}(\widehat{k})\leq\alpha\right\}$.
		\item For $s\in\mathbb{S}$, reject the null hypotheses for which $\widehat{p}_{\text{weighted}}(s, \widehat{k})\leq \widehat{p}^{(l_{\widehat{k}})}_{\text{weighted}}(\widehat{k})$.
	\end{enumerate}
	\caption{The data-driven STRAW procedure}
	\label{alg:2}
\end{algorithm}

\section{Simulation Studies}
\label{sec-3}

\par
In this section, we carry out comprehensive simulation studies to investigate the numerical performance of the STRAW procedure. According to the dimension of the simulated data, the simulations are divided into two parts: the one-dimensional setting and the two-dimensional setting. We compare the oracle and data-driven STRAW procedures (STRAW.or and STRAW.dd) with some state-of-the-art multiple testing procedures:  (1) the BH procedure (BH, Benjamini \& Hochberg, 1995); and (2) the oracle and data-driven LAWS procedures  (LAWS.or and LAWS.dd, Cai {\it et al.,} 2022). An R package, termed as STRAW, is developed for implementing the STRAW procedure and it is provided in https://github.com/wpf19890429/STRAW. All simulation results are based on $100$ repetitions and the FDR level is fixed at $0.1$. Without loss of generality, we choose $B_1=0.5$, $B_2=5$ and $L=18$.

\subsection{Simulation Studies in One-Dimensional Settings}\label{sec-3.1}

\par
In this section, the simulated data are generated from one-dimensional settings. Specifically, consider simultaneous testing of $m=5000$ null hypotheses located at a region of one-dimensional Euclidean space. The local sparsity levels ($\{\pi_1(s)\}^m_{s=1}$) are set to
\begin{eqnarray*}
	\pi_1(s) &=& \pi_{1d}, \text{~for~} s\in\{1001, 1002, \cdots, 1200\}\cup\{2001, 2002, \cdots, 2200\}\cup\\
	&&~~~~~~~~~~~~~~~~\{3001, 3002, \cdots, 3200\}\cup\{4001, 4002, \cdots, 4200\};\\
	\pi_1(s) &=& 0.01, \text{~for~other~locations}.
\end{eqnarray*}
The underlying states of the null hypotheses ($\{\theta(s)\}^m_{s=1}$) and the corresponding summary statistics ($\{X(s)\}^m_{s=1}$) are generated from the following models:
\begin{eqnarray*}
	\theta(s)          &\sim& \text{Bernoulli}(\pi_1(s)), \text{~for~} s=1,\cdots,m,\\
	X(s)\mid\theta(s)  &\sim& (1-\theta(s))N(0, 1) + \theta(s)N(\mu, 1), \text{~for~} s=1,\cdots,m.
\end{eqnarray*}
Then the $p$-values with respect to $\{X(s)\}^m_{s=1}$ can be calculated by using
\[
p(s) = 2 \left(1-\Phi(|X(s)|)\right),
\]
where $\Phi(\cdot)$ is the CDF of the standard normal distribution.

\par
The simulations are conducted in the following two Scenarios.

\par
{\bf Scenario 1:} fix $\pi_{1d}=0.6$ and change $\mu$ from $1.5$ to $2.0$.

\par
{\bf Scenario 2:} fix $\mu=2.0$ and change $\pi_{1d}$ from $0.4$ to $0.6$.

\par
The simulation results are presented in Figure 1. Panels (a) and (c) of Figure 1 demonstrate the following findings:
\begin{enumerate}
	\item The FDR values of all the comparison methods can be effectively controlled below $0.1$.
	\item Both the oracle STRAW procedure and the oracle LAWS procedure exhibit better FDR control in the vicinity of $0.1$.
	\item The BH procedure, which does not consider spatial information, displays a somewhat conservative behavior.
\end{enumerate}
Moreover, both the data-driven STRAW and the data-driven LAWS procedures yield FDR values lower than $0.1$ due to some level of conservatism in estimating ${\pi_1(s)}^m_{s=1}$. From Panels (b) and (d) of Figure 1, we can observe that the oracle STRAW procedure has the largest average number of true positives (ATP) value, followed by the oracle LAWS procedure, the data-driven STRAW procedure, the data-driven LAWS procedure, and the BH procedure. This indicates that the proposed STRAW procedure can indeed improve the efficacy of spatial multiple testing in one-dimensional setting. Moreover, the larger the value of $\mu$, the stronger the signal, so the ATP values of each method increase as $\mu$ increases. Note also that the larger the value of $\pi_{1d}$, the more signals there are. It is easy to understand that the ATP value of each method increases as $\pi_{1d}$ increases. We also conduct simulations to evaluate the estimation of the local sparsity level via the kernel-smooth Lfdr statistic. The corresponding simulation results are displayed in Appendix.

\begin{figure}[htbp]
	\centering
	\begin{subfigure}{0.48\textwidth}
		\includegraphics[width=3in]{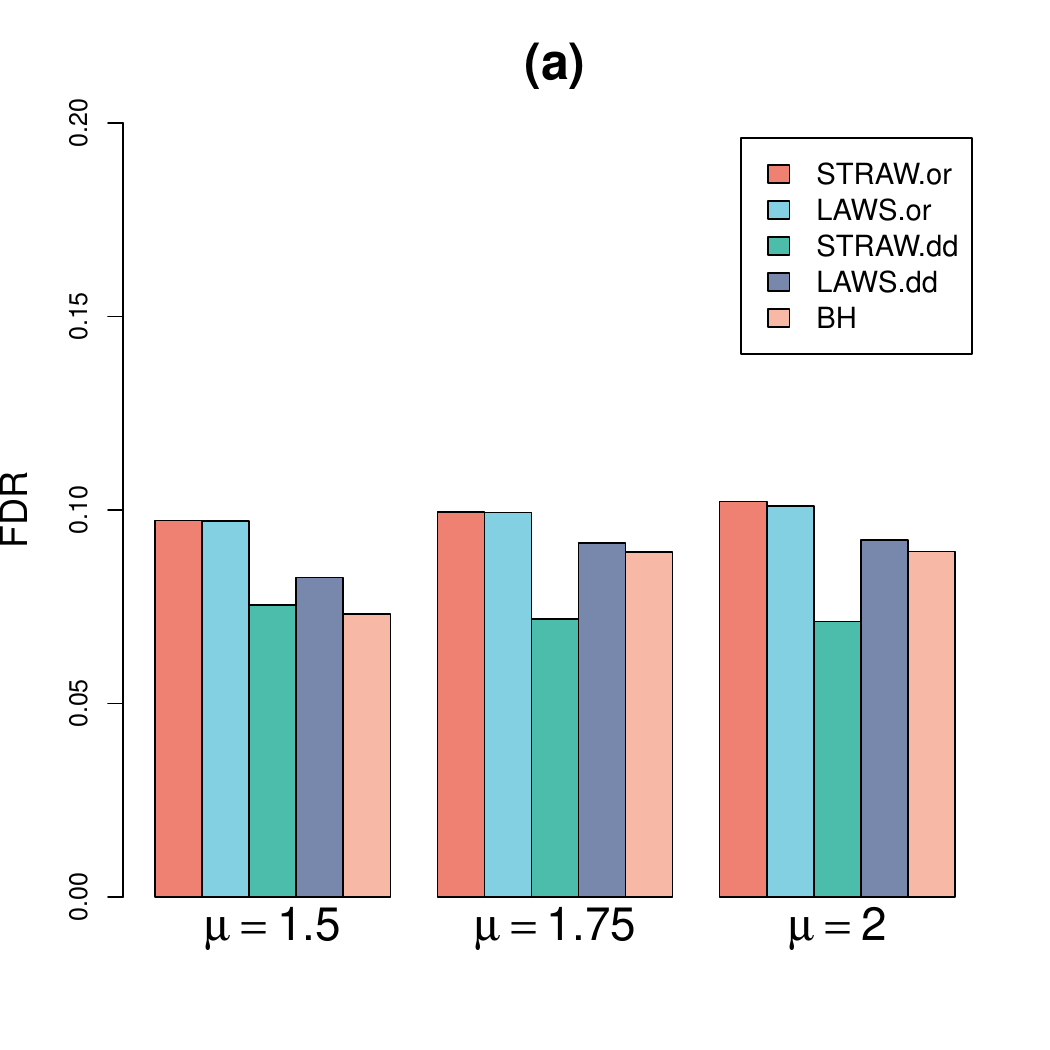}
	\end{subfigure}
\begin{subfigure}{0.48\textwidth}
	\includegraphics[width=3in]{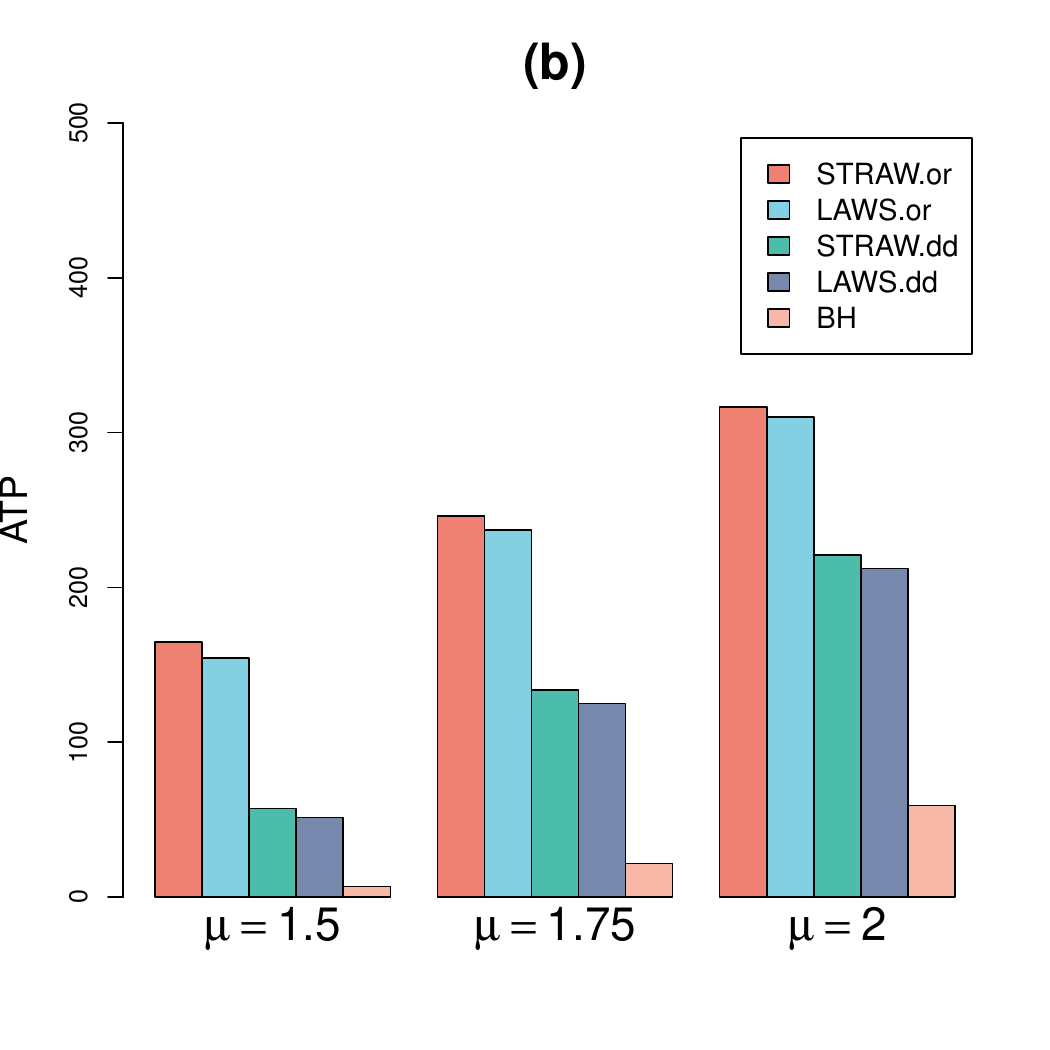}
\end{subfigure}\\
\begin{subfigure}{0.48\textwidth}
	\includegraphics[width=3in]{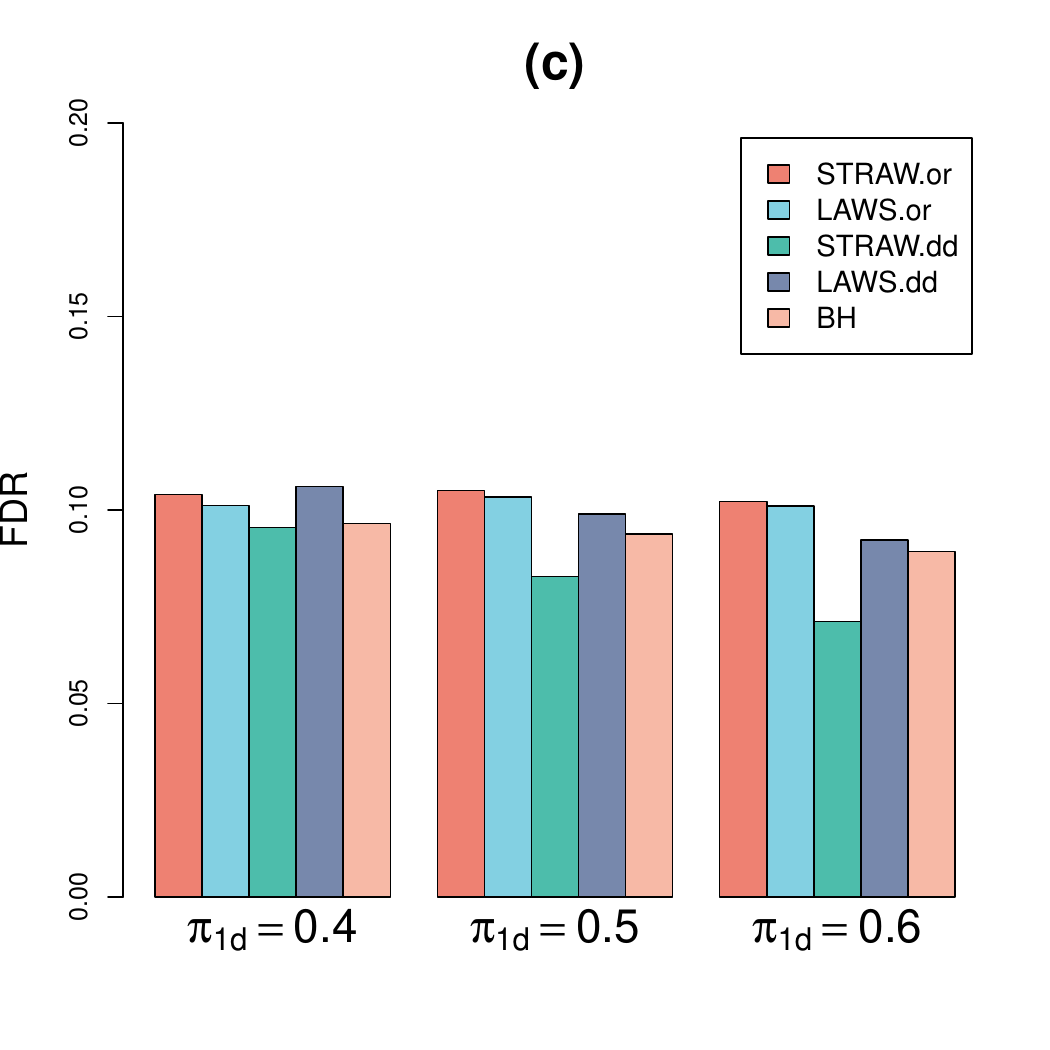}
\end{subfigure}
\begin{subfigure}{0.48\textwidth}
	\includegraphics[width=3in]{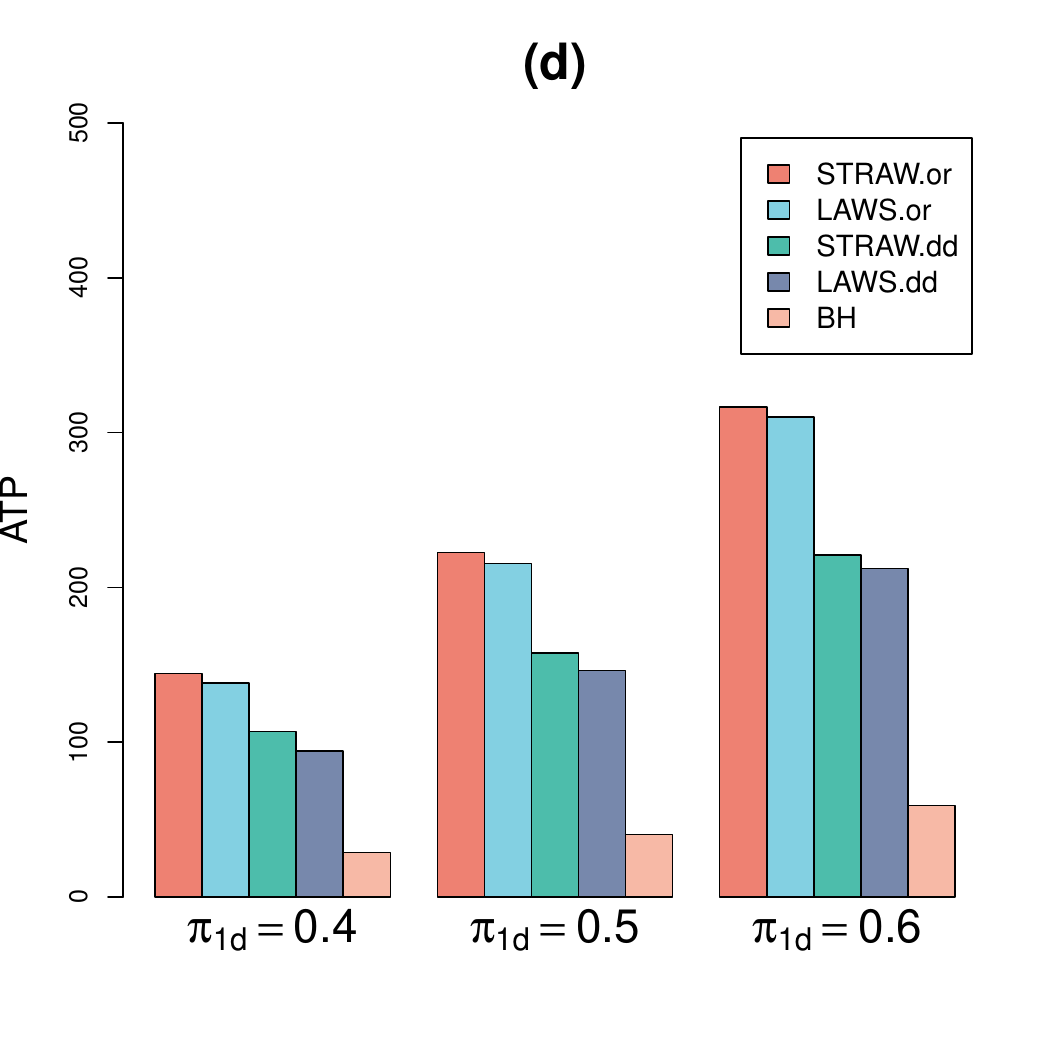}
\end{subfigure}
	\centering
	\caption{\footnotesize Simulation results in one-dimensional settings: (a)-(b) simulation results in Scenario 1; (c)-(d) simulation results in Scenario 2.}
	\label{fig:1}
\end{figure}

\subsection{Simulation Studies in Two-Dimensional Settings}

\par
In this section, consider simultaneous testing of $m=6400$ null hypotheses located at a $80\times80$ lattice in a region of two-dimensional Euclidean space. The local sparsity levels $\{\pi_1(s)\}_{s\in\{1,\cdots,80\}\times\{1,\cdots,80\}}$ are set to
\begin{eqnarray*}
	\pi_1(s) &=& \pi_{2d}, \text{~for~} s\in\{51, 52, \cdots, 65\}\times\{51, 52, \cdots, 65\}\cup\\
	&&~~~~~~~~~~~~~~~~\{s\mid d(s, (20, 20))\leq 10\};\\
	\pi_1(s) &=& 0.01, \text{~for~other~locations},
\end{eqnarray*}
where $d(x,y)$ is the Euclidean distance between $x$ and $y$. Similar to one-dimensional setting, let $\{\theta(s)\}_{s\in\{1,\cdots,80\}\times\{1,\cdots,80\}}$ and $\{X(s)\}_{s\in\{1,\cdots,80\}\times\{1,\cdots,80\}}$ be generated from
\begin{eqnarray*}
	\theta(s)          &\sim& \text{Bernoulli}(\pi_1(s)), \text{~for~} s\in\{1,\cdots,80\}\times\{1,\cdots,80\},\\
	X(s)\mid\theta(s)  &\sim& (1-\theta(s))N(0, 1) + \theta(s)N(\mu, 1), \text{~for~} s\in\{1,\cdots,80\}\times\{1,\cdots,80\}.
\end{eqnarray*}
Consider the following two scenarios.

\par
{\bf Scenario 3:} fix $\pi_{2d}=0.6$ and change $\mu$ from $1.5$ to $2.0$.

\par
{\bf Scenario 4:} fix $\mu=2.0$ and change $\pi_{2d}$ from $0.4$ to $0.6$.

\par
The corresponding simulation results are displayed in Figure 2. We can see from Figure 2 that all methods control the FDR well with one exception: the FDR of the data-driven STRAW procedure slightly exceeds $0.1$ when $\pi_{2d}=0.4$. By incorporating the spatial information successfully, both the STRAW procedures (STRAW.or and STRAW.dd) and the LAWS procedures (LAWS.or and LAWS.dd) exhibit significantly higher statistical power compared to the BH procedure. Moreover, the STRAW procedure outperforms the LAWS procedure in terms of power due to the more flexible selection of weights. Additionally, we conduct simulation studies in three-dimensional settings, and the comprehensive simulation results are provided in Appendix.

\begin{figure}[htp]
	\centering
	\begin{subfigure}{0.48\textwidth}
		\includegraphics[width=3in]{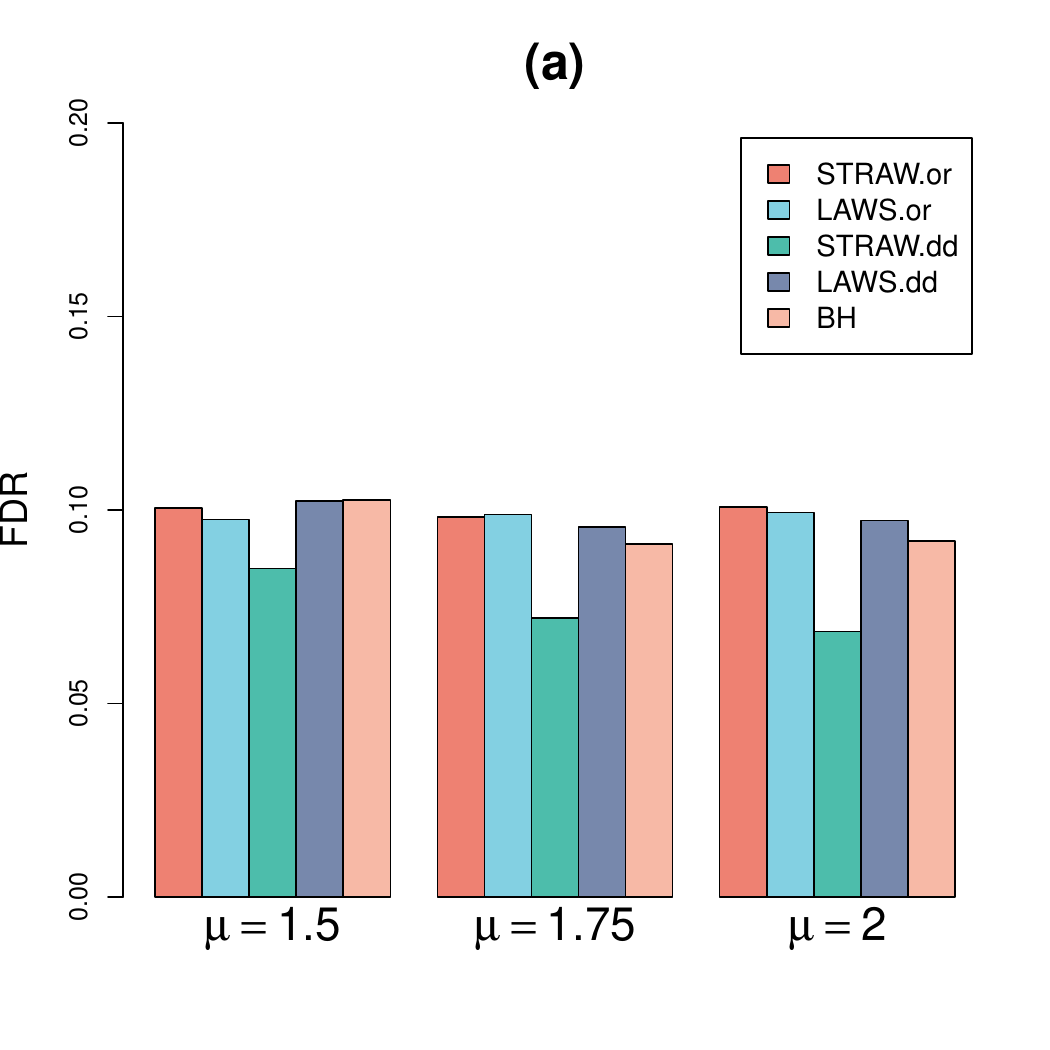}
	\end{subfigure}
\begin{subfigure}{0.48\textwidth}
	\includegraphics[width=3in]{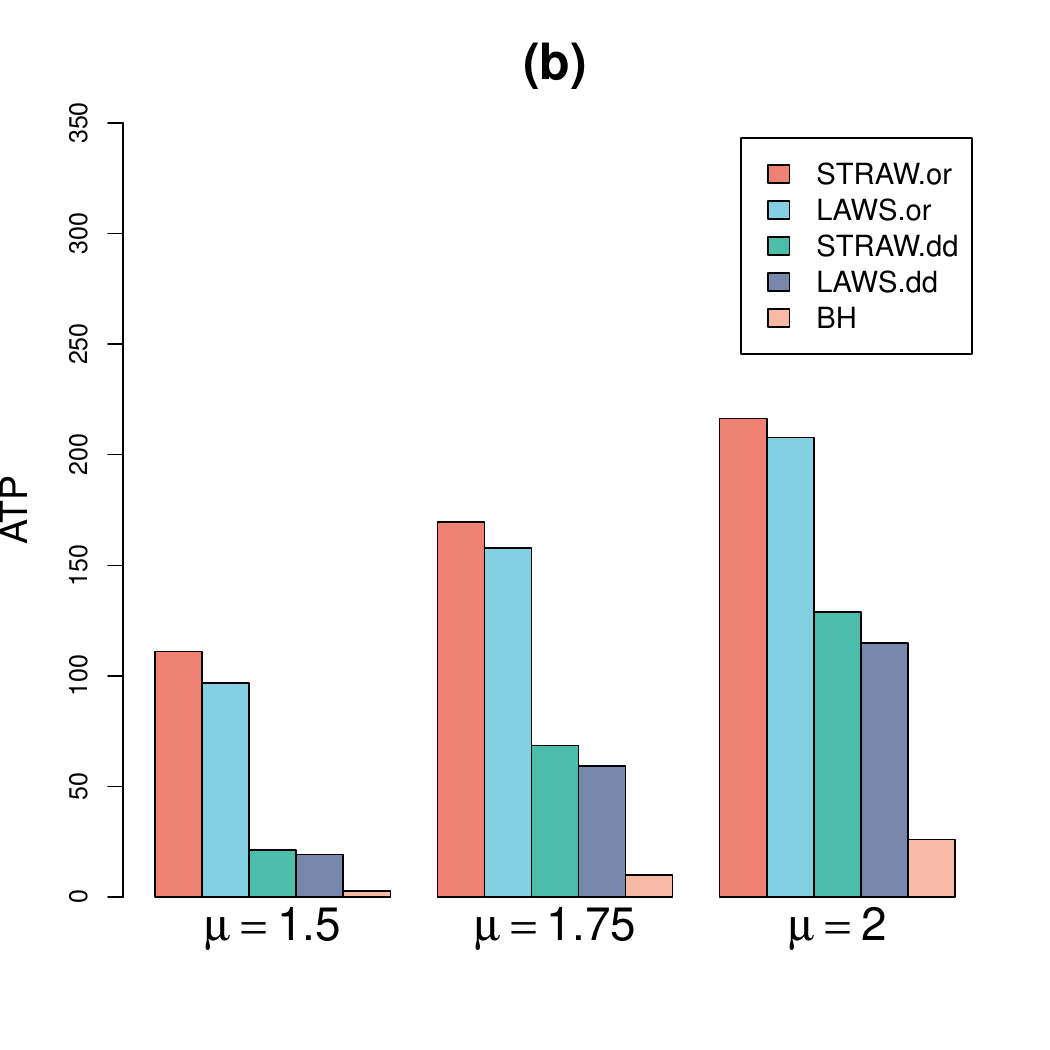}
\end{subfigure}\\
\begin{subfigure}{0.48\textwidth}
	\includegraphics[width=3in]{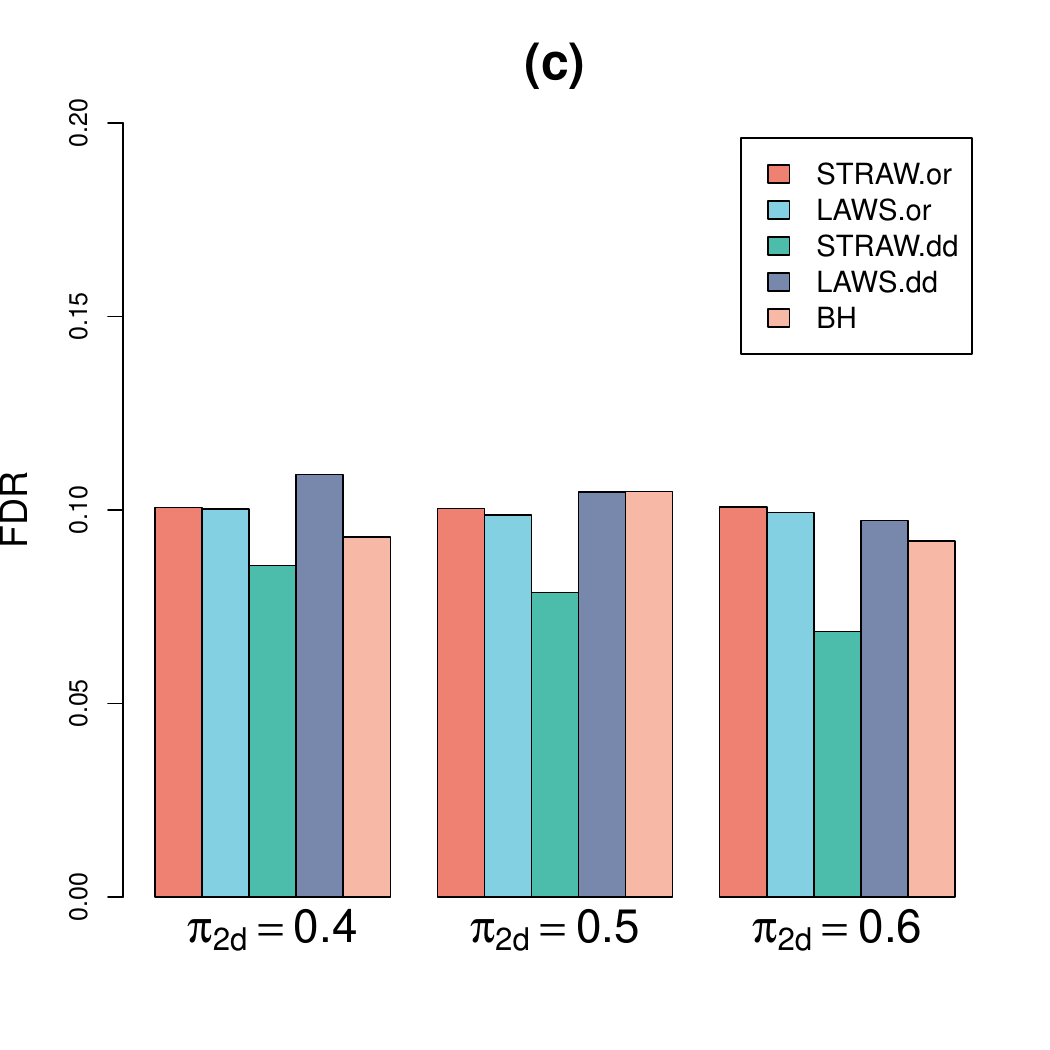}
\end{subfigure}
\begin{subfigure}{0.48\textwidth}
	\includegraphics[width=3in]{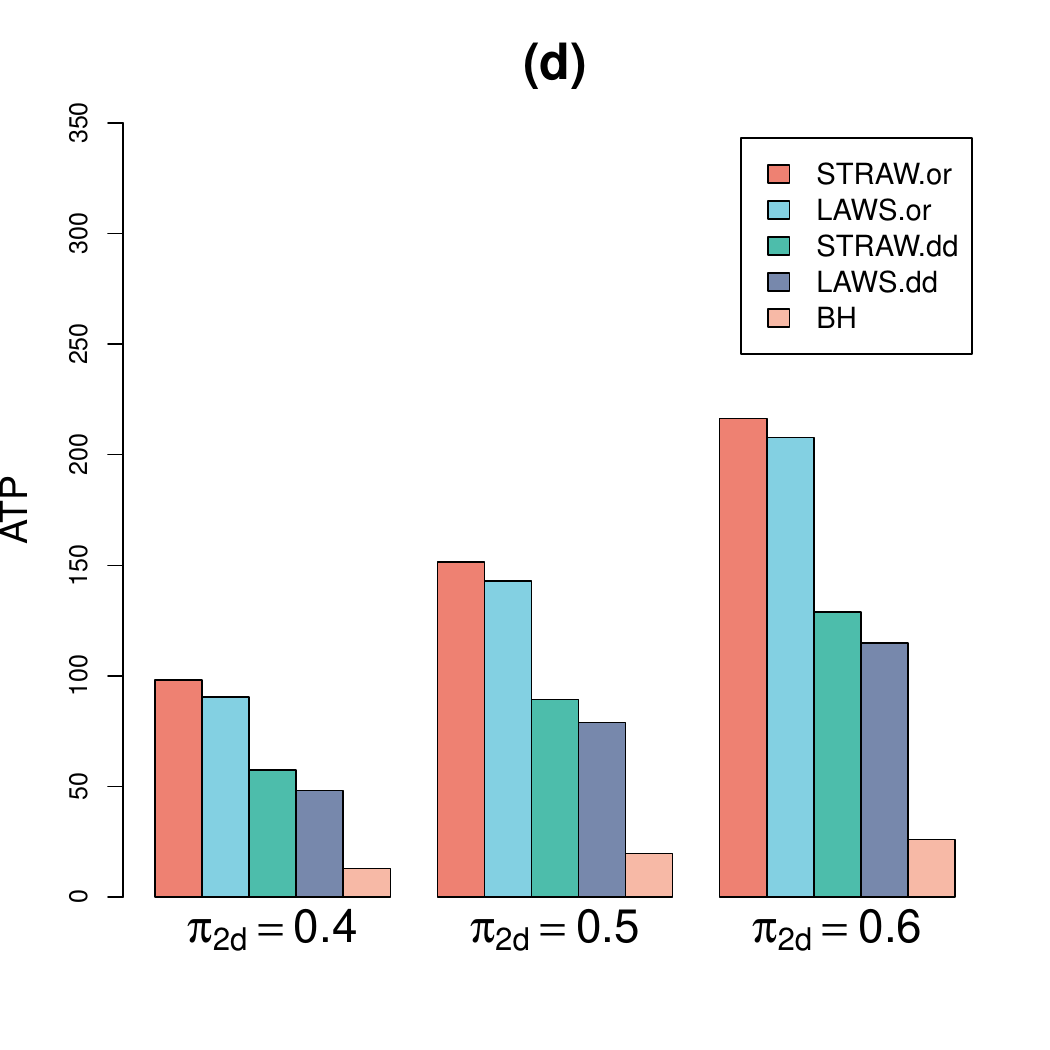}
\end{subfigure}	
	\centering
	\caption{  \footnotesize Simulation results in two-dimensional settings: (a)-(b) simulation results in Scenario 3; (c)-(d) simulation results in Scenario 4.}
	\label{fig:4}
\end{figure}

\section{Real Data Analysis}

\par
To demonstrate the effectiveness of the STRAW procedure in practical applications, we employ it in a study of attention deficit hyperactivity disorder (ADHD), utilizing magnetic resonance imaging (MRI) data. The corresponding data is produced by the ADHD-200 Sample Initiative, and can be accessed at http://neurobureau.projects.nitrc.org/ADHD200/Data.html. It has been reported that patients with ADHD had decreased amplitude of low-frequency ($0.01-0.08$ Hz) fluctuation (ALFF) in some brain regions, including the right inferior frontal cortex, left sensorimotor cortex, and bilateral cerebellum and the vermis (Zang {\it et al.,} 2007). However, it has been indicated that the ALFF is sensitive to the physiological noise (Zou {\it et al.,} 2008). To overcome this limitation to some extent, we employ the fractional ALFF (fALFF) proposed by Zou {\it et al.} (2008) for this real data analysis. Specifically, this dataset comprises $194$ individuals, with $78$ individuals diagnosed with ADHD and 116 are normal controls, and is preprocessed to generate the fALFF by the Neuro Bureau. To identify brain regions showing significant differences between individuals with ADHD and controls, we perform two-sample $t$-tests to compare the two groups. The $p$-values are calculated by using the normal approximation, and we employ the LAWS procedure as a comparative method.

\par
The detailed testing results are displayed in Figures 3 and 4. In total, $40461$ out of $133574$ ($49\times58\times47$) valid brain regions are tested simultaneously. Figure 3 shows that the significant brain regions identified by LAWS and STRAW, respectively, with the FDR level fixed at $0.05$. We can see from Figure 3 that the significant brain regions identified by STRAW is more aggregated. Moreover, the STRAW procedure identifies $735$ significant brain regions while the LAWS procedure recovers $580$. Figure 4 display the the number of discoveries identified by each multiple testing procedures relative to the target FDR level that varies from $0.02$ to $0.05$. We can conclude that the STRAW procedure achieves a higher multiple testing efficacy by leveraging spatial structure information via a flexible weighting strategy.

\begin{figure}[htp]
	\centering
	\begin{subfigure}{0.48\textwidth}
		\includegraphics[width=3in]{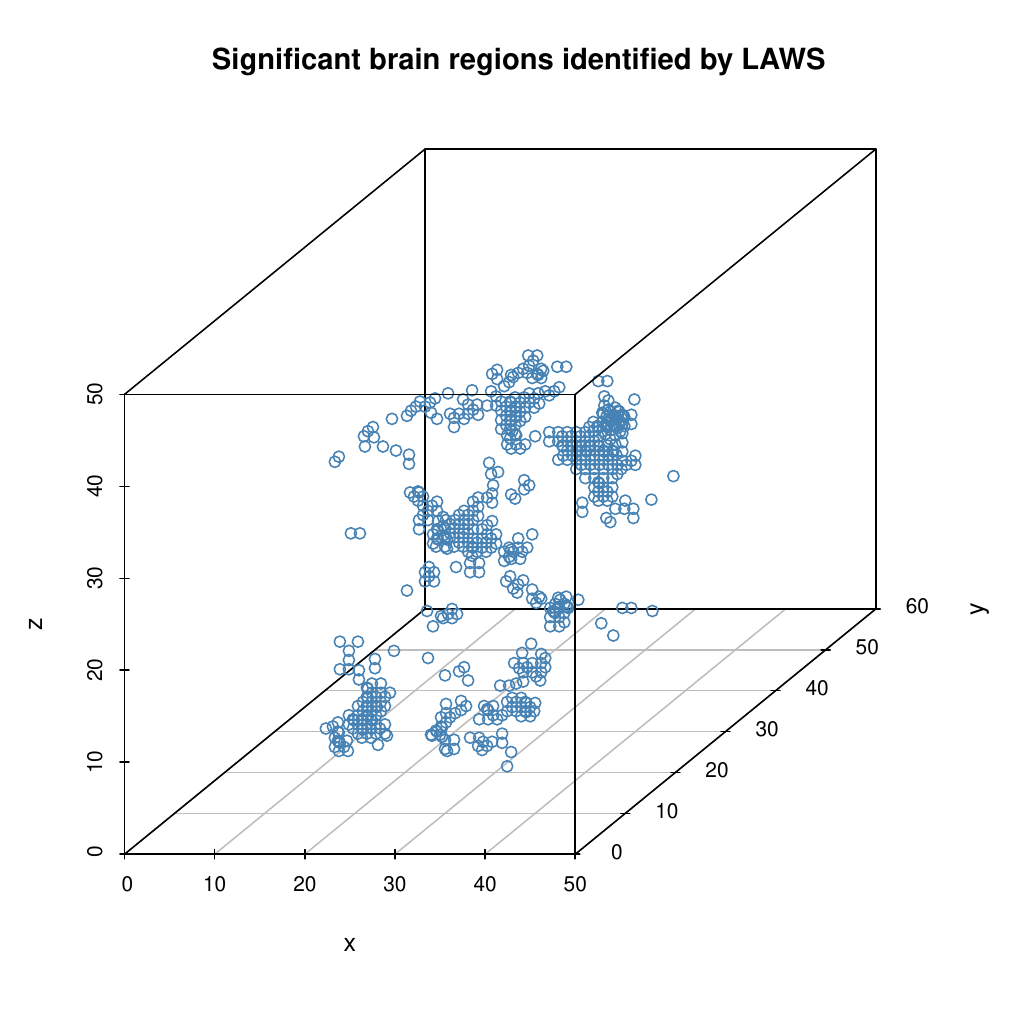}
	\end{subfigure}
	\begin{subfigure}{0.48\textwidth}
		\includegraphics[width=3in]{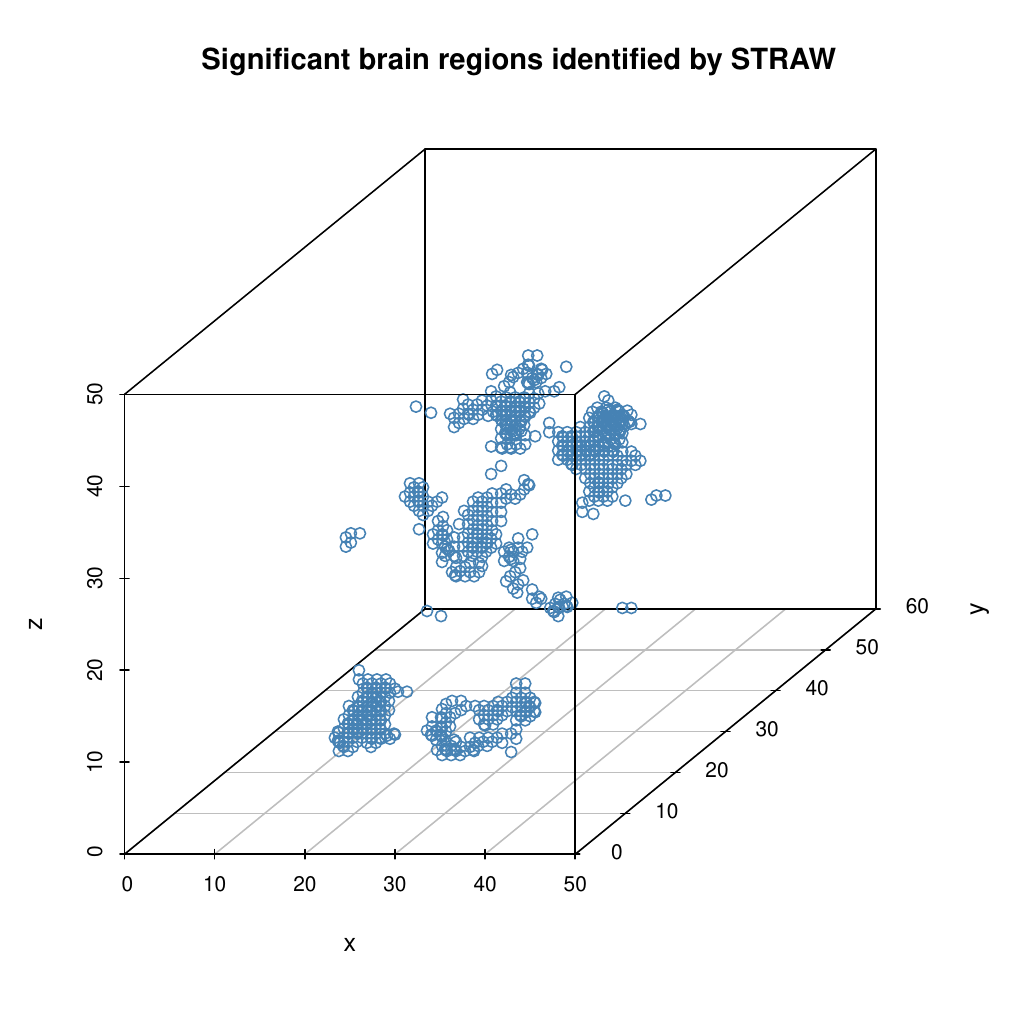}
	\end{subfigure}
	\centering
	\caption{  \footnotesize Significant brain regions identified by LAWS and STRAW, respectively, with the FDR level fixed at $0.05$.}
	\label{fig:4}
\end{figure}

\begin{figure}[H]
	\centering
	\includegraphics[height=3.5in,width=4in]{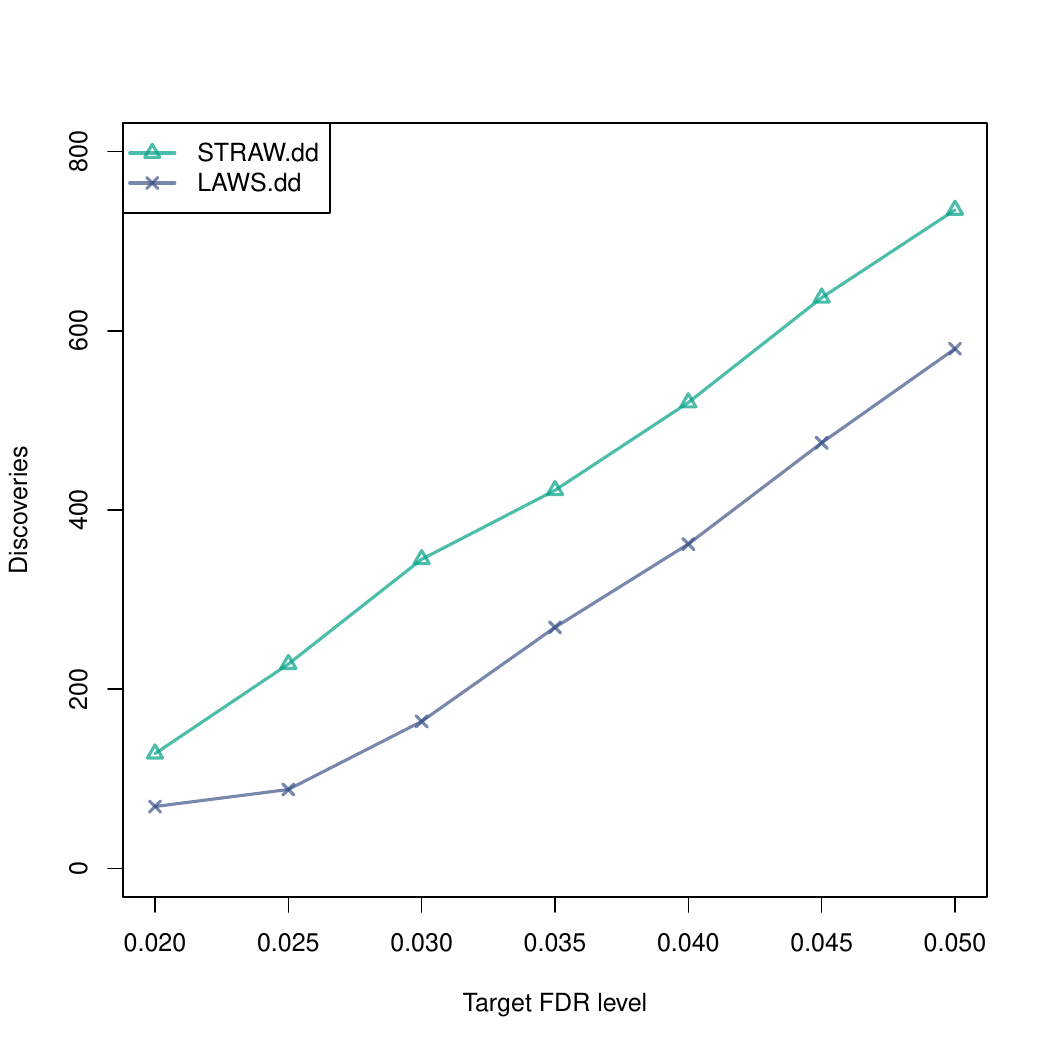}\\[-5mm]
	\caption{\footnotesize{The number of discoveries identified by LAWS and STRAW relative to the target FDR level that varies from $0.02$ to $0.05$.}}
	\label{fig:4}
\end{figure}

\section{Discussion}

\par
In this paper, we present a novel large-scale spatial multiple testing procedure based on a new class of weighted $p$-values. One key contribution of our approach lies in the incorporation of spatial information into multiple testing through the use of a class of weighted $p$-values. By assigning appropriate weights to the individual $p$-values based on their spatial proximity, the proposed approach is able to account for spatial structure information and enhance the accuracy of the testing results. One important aspect discussed in this paper is the selection of suitable turning parameter $k$ for the spatial weighting scheme. A fully data-driven method with the aim of rejecting as many null hypotheses as possible is provided to select $k$. Another major contribution of this paper is providing a flexible way to estimate local sparsity level by using smoothed Lfdr. The kernel function used in smoothed Lfdr can be chosen flexibly on the basis of the prior knowledge of the spatial structure. The effectiveness of the proposed multiple testing procedure has been demonstrated through extensive simulations and a real data analysis.

\par
It is worth noting that our proposed method is not limited to a specific application domain and can be applied to various fields that involve spatial data analysis, such as environmental monitoring, epidemiology, and geospatial statistics. The flexibility and effectiveness of the proposed approach make it a promising tool for researchers in these fields.

\par
In conclusion, this paper has introduced a new class of weighted $p$-values and demonstrated their application in a large-scale spatial multiple testing framework. The incorporation of spatial information through weighted $p$-values has shown significant improvements in FDR control and power. Future research can focus on exploring additional applications of the proposed method and further refining the weighting scheme to optimize its performance in different spatial data settings.

\section*{Funding}
This work is partially supported by the National Natural Science Foundation of China Grant (No. 12301333).


\newpage

\section{Appendix}

\subsection*{A.1 Proof of Theorem 1}

\par
We first introduce an algorithm (Algorithm 3) and show that it is equivalent to the oracle STRAW procedure.
\begin{algorithm}[htp]
	\begin{enumerate} \setlength\itemsep{-0.2em}
		\item The FDP is estimated by
		\[
		\widehat{\mathrm{FDP}}(t, \widetilde{k}) = \dfrac{\sum_{s\in\mathbb{S}}(1-\pi_1(s))^{1-1/\widetilde{k}} \pi_1(s)^{1/\widetilde{k}} t}{\max\left\{ \sum_{s\in\mathbb{S}}I {(p_{\text{weighted}}(s, \widetilde{k})\leq t)}, 1 \right\}}.
		\]
		\item For $0\leq\alpha\leq1$, let $\widehat{t}_{\widetilde{k}}=\sup\left\{t\geq0: \widehat{\mathrm{FDP}}(t, \widetilde{k})\leq \alpha\right\}$.
		\item For $s\in\mathbb{S}$, reject the null hypotheses for which $p_{\text{weighted}}(s, \widetilde{k})\leq \widehat{t}_{\widetilde{k}}$.
	\end{enumerate}
	\caption{An equivalent STRAW procedure}
	\label{alg:3}
\end{algorithm}

\begin{lemma}
	Algorithm 3 is equivalent to the oracle STRAW procedure.
\end{lemma}

\begin{proof}[\bf Proof of Lemma 1]
	Recall that the oracle STRAW procedure consists of two steps, the first step is to select a threshold according to
	\[
	l_{\widetilde{k}}=\max\left\{j: \frac{1}{j}\sum_{s\in\mathbb{S}} (1-\pi_1(s))^{1-1/\widetilde{k}} \pi_1(s)^{1/\widetilde{k}}p^{(j)}_{\text{weighted}}(\widetilde{k})\leq\alpha\right\}, \eqno{(\mathrm{A}.1)}
	\]
	and the second step is to reject the null hypotheses for which $p_{\text{weighted}}(s, \widetilde{k})\leq p^{(l_{\widetilde{k}})}_{\text{weighted}}(\widetilde{k})$. To prove Lemma 1, it suffices to show that $p^{(l_{\widetilde{k}})}_{\text{weighted}}(\widetilde{k})\leq \widehat{t}_{\widetilde{k}} < p^{(l_{\widetilde{k}}+1)}_{\text{weighted}}(\widetilde{k})$. By the definition of $l_{\widetilde{k}}$, it is clear that
	\[
	\dfrac{\sum_{s\in\mathbb{S}}(1-\pi_1(s))^{1-1/\widetilde{k}} \pi_1(s)^{1/\widetilde{k}} p^{(l_{\widetilde{k}})}_{\text{weighted}}(\widetilde{k})}{\max\left\{ \sum_{s\in\mathbb{S}}I {(p_{\text{weighted}}(s, \widetilde{k})\leq p^{(l_{\widetilde{k}})}_{\text{weighted}}(\widetilde{k}))}, 1 \right\}} \leq \alpha, \eqno{(\mathrm{A}.2)}
	\]
	and
	\[
	\dfrac{\sum_{s\in\mathbb{S}}(1-\pi_1(s))^{1-1/\widetilde{k}} \pi_1(s)^{1/\widetilde{k}} p^{(l_{\widetilde{k}}+1)}_{\text{weighted}}(\widetilde{k})}{\max\left\{ \sum_{s\in\mathbb{S}}I {(p_{\text{weighted}}(s, \widetilde{k})\leq p^{(l_{\widetilde{k}}+1)}_{\text{weighted}}(\widetilde{k}))}, 1 \right\}} > \alpha. \eqno{(\mathrm{A}.3)}
	\]
	By the definition of $\widehat{t}_{\widetilde{k}}$, there exists a sequence $\left\{t_l\right\}$ with $t_l\leq \widehat{t}_{\widetilde{k}}$ and $t_l\rightarrow \widehat{t}_{\widetilde{k}}$, such that
	\[
	\dfrac{\sum_{s\in\mathbb{S}}(1-\pi_1(s))^{1-1/\widetilde{k}} \pi_1(s)^{1/\widetilde{k}} t_l}{\max\left\{ \sum_{s\in\mathbb{S}}I {(p_{\text{weighted}}(s, \widetilde{k})\leq t_l)}, 1 \right\}} \leq \alpha. \eqno{(\mathrm{A}.4)}
	\]
	Note that $t_l\leq \widehat{t}_{\widetilde{k}}$ implies that $\sum_{s\in\mathbb{S}}I {(p_{\text{weighted}}(s, \widetilde{k})\leq t_l)} \leq \sum_{s\in\mathbb{S}}I {(p_{\text{weighted}}(s, \widetilde{k})\leq \widehat{t}_{\widetilde{k}})}$. It follows that
	\[
	\dfrac{\sum_{s\in\mathbb{S}}(1-\pi_1(s))^{1-1/\widetilde{k}} \pi_1(s)^{1/\widetilde{k}} t_l}{\max\left\{ \sum_{s\in\mathbb{S}}I {(p_{\text{weighted}}(s, \widetilde{k})\leq \widehat{t}_{\widetilde{k}})}, 1 \right\}} \leq \alpha. \eqno{(\mathrm{A}.5)}
	\]
	Letting $l \rightarrow+\infty$, we have
	\[
	\dfrac{\sum_{s\in\mathbb{S}}(1-\pi_1(s))^{1-1/\widetilde{k}} \pi_1(s)^{1/\widetilde{k}} \widehat{t}_{\widetilde{k}}}{\max\left\{ \sum_{s\in\mathbb{S}}I {(p_{\text{weighted}}(s, \widetilde{k})\leq \widehat{t}_{\widetilde{k}})}, 1 \right\}} \leq \alpha. \eqno{(\mathrm{A}.6)}
	\]
	Combining (A.2), (A.3) and the definition of $\widehat{t}_{\widetilde{k}}$, we conclude that
	\[
	p^{(l_{\widetilde{k}})}_{\text{weighted}}(\widetilde{k})\leq \widehat{t}_{\widetilde{k}} < p^{(l_{\widetilde{k}}+1)}_{\text{weighted}}(\widetilde{k}).
	\]
	This completes the proof of Lemma 1.
\end{proof}

\begin{algorithm}[htp]%
	\begin{enumerate} \setlength\itemsep{0.2em}
		\item Rank the rescaled weighted $p$-values from the smallest to largest $p^{*(1)}(\widetilde{k}), \cdots, p^{*(m)}(\widetilde{k})$.
		\item For $0\leq\alpha\leq1$, let $l^*_{\widetilde{k}}=\max\left\{1\leq i\leq m: p^{*(i)}(\widetilde{k}) \leq \alpha i/m \right\}$.
		\item If $l^*_{\widetilde{k}}$ exists, then the null hypotheses $H_0(s)$ for which $p^*(s, \widetilde{k}) \leq p^{*(l^*_{\widetilde{k}})}(\widetilde{k})$ are rejected; otherwise no hypothesis is rejected.
	\end{enumerate}
	\caption{Procedure 1}
	\label{alg:4}
\end{algorithm}

Define the rescaled weighted $p$-value as
\[
p^*(s, \widetilde{k}) = \min\left\{p(s)/\omega^*(s, \widetilde{k}), 1 \right\},
\]
where $\omega^*(s, \widetilde{k}) = m \left(\frac{\pi_1(s)}{1-\pi_1(s)}\right)^{1/\widetilde{k}} \left\{\sum_{s\in\mathbb{S}} \left(\frac{\pi_1(s)}{1-\pi_1(s)}\right)^{1/\widetilde{k}}\right\}^{-1}$, for $s\in\mathbb{S}$. Applying the BH procedure to the rescaled weighted $p$-values yields the following procedure.

Likewise, we can obtain an algorithm (Algorithm 5) which is equivalent to Procedure 1.

\begin{algorithm}[htp]%
	\begin{enumerate} \setlength\itemsep{-0.2em}
		\item The FDP is estimated by
		\[
		\widehat{\mathrm{FDP}}^*(t, \widetilde{k}) = \dfrac{m t}{\max\left\{ \sum_{s\in\mathbb{S}}I {(p^*(s, \widetilde{k})\leq t)}, 1 \right\}}.
		\]
		\item For $0\leq\alpha\leq1$, let $t^*_{\widetilde{k}}=\sup\left\{t\geq0: \widehat{\mathrm{FDP}}^*(t, \widetilde{k})\leq \alpha\right\}$.
		\item For $s\in\mathbb{S}$, reject the null hypotheses for which $p^*(s, \widetilde{k})\leq t^*_{\widetilde{k}}$.
	\end{enumerate}
	\caption{An equivalent Procedure 1}
	\label{alg:5}
\end{algorithm}

\begin{lemma}
	Algorithm 5 is equivalent to Procedure 1.
\end{lemma}
Since the proof of Lemma 2 is similar to that of Lemma 1, we have omitted it..

\begin{lemma}
	Suppose that Assumptions (A1) and (A2) hold, then we have
	\[
	\underset{m\rightarrow\infty}{\lim\sup}~\mathrm{FDR}_{Procedure 1} \leq \alpha, \text{~and~} \lim\limits_{m\rightarrow\infty} \Pr (\mathrm{FDP}_{Procedure 1} \leq \alpha+\varepsilon) = 1,
	\]
	for any $\varepsilon>0$.
\end{lemma}

\begin{proof}[\bf Proof of Lemma 3]
	Let $\{s_1, \cdots, s_m\}$ be any order arrangement of $\{s\in\mathbb{S}\}$. We first show that Procedure 1 is capable of controlling the FDR when $\{p^*(s_i, \widetilde{k})\}^m_{i=1}$ are independent and then demonstrate that it performs asymptotically the same in the dependent case as in the independent case.
	
	\leftline{\bf Independent Case:}
	
	\par
	Let $t_m=\left(2\log m-2\log\log m\right)^{1/2}$ and $z_i=\Phi^{-1}(1-p(s_i)/2)$, for $i=1, \cdots, m$. By Assumption (A5), we have, as $m\rightarrow+\infty$,
	\[
	\Pr\left\{\sum_{\theta(s_i)=0} I( z_i \geq (c\log m)^{1/2+\rho/4})\geq \{1/(\pi^{1/2}+\delta)\}(\log m)^{1/2}\right\}\rightarrow 1,
	\]
	for some constant $c>0$. Recall that $\omega^*(s_i, \widetilde{k}) = m \left(\frac{\pi_1(s)}{1-\pi_1(s)}\right)^{1/\widetilde{k}}\left\{\sum_{s\in\mathbb{S}}\left(\frac{\pi_1(s)}{1-\pi_1(s)}\right)^{1/\widetilde{k}}\right\}^{-1}$ and $\pi_1(s)\in[\xi, 1-\xi]$, we have
	\[
	\omega^*(s_i, \widetilde{k})\geq\left(\dfrac{\xi}{1-\xi}\right)^{2/\widetilde{k}} > m^{-b},
	\]
	for some constant $b>0$. Thus, for those indices $\theta(s_i)=1$ such that $ z_i \geq (c\log m)^{1/2+\rho/4}$, we have
	\begin{eqnarray*}
		p^*(s_i, \widetilde{k}) &=&    p(s_i)/\omega^*(s_i, \widetilde{k}) \\
		&\leq&  2(1-\Phi((c\log m)^{1/2+\rho/4}))/\omega^*(s_i, \widetilde{k}) \\
		&\leq&  2m^b(1-\Phi((c\log m)^{1/2+\rho/4})) \\
		&\leq&  o(m^{-M}),
	\end{eqnarray*}
	for any constant $M>0$. Let $z^*_i=\Phi^{-1}(1-p^*(s_i, \widetilde{k})/2)$, for $i=1, \cdots, m$, then we have
	\[
	\Pr\left\{\sum^m_{i=1}I(z^{*}_i\geq (2\log m)^{1/2})\geq \left\{1/(\pi^{1/2}\alpha)+\delta\right\}^{-1}(\log m)^{1/2} \right\}\rightarrow 1,
	\]
	as $m\rightarrow+\infty$. Recall the definition of $\widehat{\mathrm{FDP}}^*(t, \widetilde{k})$ in Algorithm 5, we have
	\[
	\dfrac{m t}{\max\left\{ \sum^m_{i=1}I {(p^*(s_i, \widetilde{k})\leq t)}, 1 \right\}} = \dfrac{m G(u)}{\max\left\{\sum^m_{i=1}I(z^*_i\geq u), 1 \right\}},
	\]
	where $u=\Phi^{-1}(1-t/2)$ and $G(u)=2(1-\Phi(u))$. Let $t_m=\left(2\log m -2\log\log m\right)^{1/2}$, then, with probability tending to one, we have
	\[
	\dfrac{2m}{\sum^m_{i=1}I(z^{*}_i\geq t_m)} \leq \dfrac{2m}{\sum^m_{i=1}I(z^{*}_i\geq (2\log m)^{1/2})} \leq 2m \left\{1/(\pi^{1/2}\alpha)+\delta\right\}^{-1}(\log m)^{-1/2}.
	\]
	Note that
	\begin{eqnarray*}
		1-\Phi(t_m) &\sim&  \dfrac{1}{\sqrt{2\pi}} \cdot \dfrac{1}{t_m} \cdot \exp\left\{-\dfrac{t^2_m}{2}\right\}\\
		&=&  \dfrac{1}{\sqrt{2\pi}} \cdot \dfrac{1}{(2\log m-2\log\log m)^{1/2}} \cdot \dfrac{\log m}{m}.
	\end{eqnarray*}
	Then, with probability tending to one, we have
	\begin{eqnarray*}
		\dfrac{m G(t_m)}{\sum^m_{i=1}I(z^*_i\geq t_m)} &=&  \dfrac{2m}{\sum^m_{i=1}I(z^{*}_i\geq t_m)} \cdot \left(1-\Phi(t_m)\right)\\
		&\leq& \dfrac{ 2m }{\left\{1/(\pi^{1/2}\alpha)+\delta\right\}(\log m)^{1/2}} \cdot \dfrac{\log m}{\sqrt{2\pi}(2\log m-2\log\log m)^{1/2}m}\\
		&\leq& \alpha.
	\end{eqnarray*}
	Since the threshold for $z^*_i$ is selected by $u^*=\inf\left\{u\geq0: {m G(u)}/{\max\left\{\sum^m_{i=1}I(z^*_i\geq u), 1 \right\}}\leq \alpha\right\}$, it suffices to show that, uniformly in $0\leq u\leq t_m$, there exists a constant $0<c_0\leq 1$, such that
	\[
	\left|\dfrac{\sum_{\theta(s_i)=0}I(z^*_i>u)-c_0 m_0 G(u)}{c_0 m_0 G(u)}\right|\rightarrow0, \eqno{(\mathrm{A}.7)}
	\]
	in probability, where $m_0=\left|\{s\in \mathbb{S}: \theta(s)=0\}\right|$.
	
	\par
	To control the FDR at level $\alpha$, the ideal choice of the threshold $u^{o}$ for $z^{*}_i$ is
	\[
	u^{o} =\inf \left\{u\geq 0: \dfrac{\sum_{\theta(s_i)=0} I(z^{*}_i\geq u)}{\sum^m_{i=1}I(z^{*}_i)\geq u}\leq\alpha \right\}.
	\]
	By the law of large numbers, it is clear that,
	\[
	\left| \dfrac{\sum_{\theta(s_i)=0}I(z^{*}_i\geq u)-\sum_{\theta(s_i)=0}\Pr(z^{*}_i\geq u)}{\sum_{\theta(s_i)=0}\Pr(z^{*}_i\geq u)}\right|\rightarrow 0, \text{~as~} m\rightarrow \infty. \eqno{(\mathrm{A}.8)}
	\]
	Thus we can obtain a good estimate of $u^{o}$, that is,
	\[
	\widehat{u}^{o}=\inf\left\{u\geq0, \dfrac{\sum_{\theta(s_i)=0}\Pr(z^{*}_i\geq u)}{\sum^m_{i=1}I(z^{*}_i\geq u)}\leq\alpha\right\}.
	\]
	Note that the original $p$-values under the null are uniformly distributed, which is not dependent on the spatial locations $\{s\in\mathbb{S}\}$. According to Theorem 1 in Genovese et al. (2006), Procedure 1 controls the FDR at level $\alpha m_0/m$. That is, let
	\[
	l^*_{\widetilde{k}}=\max\left\{1\leq i\leq m: p^{*(i)}(\widetilde{k}) \leq \alpha i/m \right\},
	\]
	then we have
	\[
	\mathrm{E} \left(\dfrac{\sum_{\theta(s)=0} I(p^*(s, \widetilde{k})\leq p^{*(l^*_{\widetilde{k}})}(\widetilde{k}))}{\max\{\sum^m_{i=1} I(p^*(s, \widetilde{k})\leq p^{*(l^*_{\widetilde{k}})}(\widetilde{k})), 1\}}\right)\leq \alpha m_0/m.
	\]
	By the definition of $z^{*}_i$, Procedure 1 is equivalent to reject all hypotheses with
	\[
	\dfrac{2m(1-\Phi(z^{*}_{(i)}))}{i}\leq\alpha,
	\]
	where $z^{*}_{(i)}=\Phi^{-1}(1-p^{*(i)}(\widetilde{k})/2)$. That is to find
	\[
	u^*=\inf\left\{u\geq0: \dfrac{m G(u)}{\max\left\{\sum^m_{i=1}I(z^*_i\geq u), 1 \right\}}\leq \alpha\right\},  \eqno{(\mathrm{A}.9)}
	\]
	and reject all hypotheses with $z^{*}_i\geq u^*$. Then we have
	\[
	\mathrm{E}\left(\dfrac{\sum_{\theta(s_i)=0}I(z^{*}_i\geq u^*)}{\sum^m_{i=1}I(z^{*}_i\geq u^*)}\right)\leq\alpha m_0/m.
	\]
	By (A.8) and the definitions of $u^o$ and $\widehat{u}^{o}$, we have
	\[
	\mathrm{E}\left(\dfrac{\sum_{\theta(s_i)=0}I(z^{*}_i\geq \widehat{u}^{o})}{\sum^m_{i=1}I(z^{*}_i\geq \widehat{u}^{o})}\right)\rightarrow \alpha, \text{~as~} m\rightarrow \infty.
	\]
	Hence, the procedure based on the threshold $u^*$ is is more conservative than that based on $\widehat{t}^{o}$. Thus, there exists a constant $0<c_0\leq 1$, such that, uniformly in $0\leq u\leq t_m$,
	\[
	\left|\dfrac{\sum_{\theta(s_i)=0}\Pr(z^*_i>u)-c_0 m_0 G(u)}{c_0 m_0 G(u)}\right|\rightarrow0. \eqno{(\mathrm{A}.10)}
	\]
	Combining (A.8) and (A.10), (A.7) is proved.
	
	\leftline{\bf Dependent Case:}
	Note that
	\[
	\Phi^{-1}\{1-[1-\Phi\{(c_1 \log m + c_2 \log\log m)^{1/2}\}]/\omega^*(s, \widetilde{k})\} = c_1 \log m +c_2 \log\log m +c_3,
	\]
	for some constsnt $c_1, c_2, c_3$. By Assumption (A1) and the proof of Theorem 1 in Xia et al. (2019), we have, as $m\rightarrow+\infty$,
	\[
	\left|\dfrac{\sum_{\theta(s_i)=0}(I(z^{*}_i\geq u)-\Pr(z^{*}_i\geq u))}{\sum_{\theta(s_i)=0}\Pr(z^{*}_i\geq u)}\right|\rightarrow0, \eqno{(\mathrm{A}.11)}
	\]
	in probability, uniformly in $0\leq t\leq t_m$. By (A.10) and (A.11), (A.7) is proved. This completes the proof of Lemma 3.
\end{proof}

\begin{proof}[\bf Proof of Theorem 1]
	By Lemma 1, we have that the oracle STRAW procedure is equivalent to Algorithm 3, which is to reject the null hypotheses with $p_{\text{weighted}}(s, \widetilde{k})\leq \widehat{t}_{\widetilde{k}}$ for $s\in\mathbb{S}$, where
	\[
	\widehat{t}_{\widetilde{k}}=\sup\left\{t\geq0: \dfrac{\sum_{s\in\mathbb{S}}(1-\pi_1(s))^{1-1/\widetilde{k}} \pi_1(s)^{1/\widetilde{k}} t}{\max\left\{ \sum_{s\in\mathbb{S}}I {(p_{\text{weighted}}(s, \widetilde{k})\leq t)}, 1 \right\}}\leq \alpha\right\}.
	\]
	From the definition of $\widehat{t}_{\widetilde{k}}$, it is clear that proving Theorem 1 only requires proving the conclusion that for all $t\geq\widehat{t}_{\widetilde{k}}$, there exists a constant $0<c\leq1$, as $m\rightarrow\infty$,
	\[
	\left| \dfrac{\sum_{s\in\mathbb{S}}\left[I(p_{\text{weighted}}(s, \widetilde{k})\leq t, \theta(s)=0)- c(1-\pi_1(s))^{1-1/\widetilde{k}} \pi_1(s)^{1/\widetilde{k}} t\right]}{\sum_{s\in\mathbb{S}} c(1-\pi_1(s))^{1-1/\widetilde{k}} \pi_1(s)^{1/\widetilde{k}} t} \right| \rightarrow 0,
	\]
	in probability.

	\par
	{\bf Step 1:} We first show that for all $t\geq\widehat{t}_{\widetilde{k}}$, as $m\rightarrow\infty$,
	\[
	\left| \dfrac{\sum_{\theta(s)=0} \left[I(p_{\text{weighted}}(s, \widetilde{k}) \leq t) - \Pr(p_{\text{weighted}}(s, \widetilde{k}) \leq t \mid \theta(s)=0) \right]}{\sum_{\theta(s)=0}\Pr(p_{\text{weighted}}(s, \widetilde{k}) \leq t, \theta(s)=0)} \right| \rightarrow 0, \eqno{(\mathrm{A}.12)}
	\]
	in probability. It follows from Lemma 2 that Procedure 1 is equivalent to Algorithm 5, which is to reject the null hypotheses with $p^*(s, \widetilde{k})\leq t^*_{\widetilde{k}}$ for $s\in\mathbb{S}$, where
	\[
	t^*_{\widetilde{k}}=\sup\left\{t\geq0: \dfrac{m t}{\max\left\{ \sum_{s\in\mathbb{S}}I {(p^*(s, \widetilde{k})\leq t)}, 1 \right\}}\leq \alpha\right\}.
	\]
	Note that $p^*(s, \widetilde{k}) = m^{-1} \left\{\sum_{s\in\mathbb{S}} \left(\frac{\pi_1(s)}{1-\pi_1(s)}\right)^{1/\widetilde{k}}\right\} p_{\text{weighted}}(s, \widetilde{k})$, where $m^{-1} \left\{\sum_{s\in\mathbb{S}} \left(\frac{\pi_1(s)}{1-\pi_1(s)}\right)^{1/\widetilde{k}}\right\}$ is a positive constant. Thus the weighted $p$-values $\left\{p_{\text{weighted}}(s, \widetilde{k})\right\}_{s\in\mathbb{S}}$ and the rescaled weighted $p$-values $\left\{p^*(s, \widetilde{k})\right\}_{s\in\mathbb{S}}$ are exactly in the same order. Let $\widetilde{t}=m \left\{\sum_{s\in\mathbb{S}} \left(\frac{\pi_1(s)}{1-\pi_1(s)}\right)^{1/\widetilde{k}}\right\}^{-1} t$. Then we have that
	\[
	\dfrac{m t}{\max\left\{ \sum_{s\in\mathbb{S}}I {(p^*(s, \widetilde{k})\leq t)}, 1 \right\}} = \dfrac{\sum_{s\in\mathbb{S}} \left(\frac{\pi_1(s)}{1-\pi_1(s)}\right)^{1/\widetilde{k}}\widetilde{t}}{\max\left\{ \sum_{s\in\mathbb{S}}I {(p_{\text{weighted}}(s, \widetilde{k})\leq \widetilde{t})}, 1 \right\}}.
	\]
	Therefore, in Procedure 1, the threshold values for $p_{\text{weighted}}(s, \widetilde{k})$ and $p^*(s, \widetilde{k})$ are
	\[
	t^1_{\widetilde{k}}=\sup\left\{t\geq0: \dfrac{\sum_{s\in\mathbb{S}} \left(\frac{\pi_1(s)}{1-\pi_1(s)}\right)^{1/\widetilde{k}}t}{\max\left\{ \sum_{s\in\mathbb{S}}I {(p_{\text{weighted}}(s, \widetilde{k})\leq t)}, 1 \right\}}\leq \alpha\right\},
	\]
	and
	\[
	t^{1*}_{\widetilde{k}}=m^{-1} \sum_{s\in\mathbb{S}} \left(\frac{\pi_1(s)}{1-\pi_1(s)}\right)^{1/\widetilde{k}}t^1_{\widetilde{k}},
	\]
	respectively. Recalling the definition of $\widehat{t}_{\widetilde{k}}$, it is clear that $t^1_{\widetilde{k}}\leq \widehat{t}_{\widetilde{k}}$. By Lemma 3, we can conclude that for all $t\geq t^{1*}_{\widetilde{k}}$, as $m\rightarrow\infty$,
	\[
	\left| \dfrac{\sum_{\theta(s)=0} \left[I(p^*(s, \widetilde{k}) \leq t) - \Pr(p^*(s, \widetilde{k}) \leq t \mid \theta(s)=0) \right]}{\sum_{\theta(s)=0}\Pr(p^*(s, \widetilde{k}) \leq t \mid \theta(s)=0)} \right| \rightarrow 0,
	\]
	in probability. This implies that for all $t\geq t^{1}_{\widetilde{k}}$, as $m\rightarrow\infty$,
	\[
	\left| \dfrac{\sum_{\theta(s)=0} \left[I(p_{\text{weighted}}(s, \widetilde{k}) \leq t) - \Pr(p_{\text{weighted}}(s, \widetilde{k}) \leq t \mid \theta(s)=0) \right]}{\sum_{\theta(s)=0}\Pr(p_{\text{weighted}}(s, \widetilde{k}) \leq t \mid \theta(s)=0)} \right| \rightarrow 0,
	\]
	in probability. Hence (A.12) is proved.
	
	\par
	{\bf Step 2:} Next we show that, as $m\rightarrow\infty$,
	\[
	\left|\dfrac{\sum_{\theta(s)=0}\Pr(p_{\text{weighted}}(s, \widetilde{k}) \leq t \mid \theta(s)=0) - \sum_{s\in\mathbb{S}}\Pr(p_{\text{weighted}}(s, \widetilde{k}) \leq t, \theta(s)=0)}{\sum_{s\in\mathbb{S}}\Pr(p_{\text{weighted}}(s, \widetilde{k}) \leq t, \theta(s)=0)} \right| \rightarrow 0, \eqno{(\mathrm{A}.13)}
	\]
	in probability. It follows from Assumption (A2) that
	\begin{eqnarray*}
		&&  \mathrm{E} \left(\left| \dfrac{\sum_{\theta(s)=0} \Pr(p_{\text{weighted}}(s, \widetilde{k}) \leq t \mid \theta(s)=0) - \sum_{s\in\mathbb{S}}\Pr(p_{\text{weighted}}(s, \widetilde{k}) \leq t, \theta(s)=0)}{\sum_{s\in\mathbb{S}}\Pr(p_{\text{weighted}}(s, \widetilde{k}) \leq t, \theta(s)=0)} \right|^2\right) \\
		&=& \mathrm{E} \left(\left| \dfrac{\sum_{s\in\mathbb{S}}\Pr(p_{\text{weighted}}(s, \widetilde{k}) \leq t \mid \theta(s)=0)\left[ I(\theta(s)=0)-\Pr(\theta(s)=0) \right]}{\sum_{s\in\mathbb{S}}\Pr(p_{\text{weighted}}(s, \widetilde{k}) \leq t, \theta(s)=0)} \right|^2\right)\\
		&=& \dfrac{\mathrm{Var}\left(\sum_{s\in\mathbb{S}}\Pr(p_{\text{weighted}}(s, \widetilde{k}) \leq t \mid \theta(s)=0)I(\theta(s)=0)\right)}{\left(\sum_{s\in\mathbb{S}}\Pr(p_{\text{weighted}}(s, \widetilde{k}) \leq t, \theta(s)=0)\right)^2 }\\
		&=& O(m^{\zeta-1}),
	\end{eqnarray*}
	where $0\leq\zeta<1$. Hence (A.13) is proved.
	
	\par
	{\bf Step 3:} Finally, we show that for all $t\geq\widehat{t}_{\widetilde{k}}$, as $m\rightarrow\infty$,
	\[
	\left| \dfrac{\sum_{s\in\mathbb{S}} \left[I(p_{\text{weighted}}(s, \widetilde{k}) \leq t, \theta(s)=0) - \Pr(p_{\text{weighted}}(s, \widetilde{k}) \leq t, \theta(s)=0)\right]}{\sum_{s\in\mathbb{S}}\Pr(p_{\text{weighted}}(s, \widetilde{k}) \leq t, \theta(s)=0)} \right| \rightarrow 0, \eqno{(\mathrm{A}.14)}
	\]
	in probability. Combining (A.1) and (A.2), we have that for all $t\geq\widehat{t}_{\widetilde{k}}$, as $m\rightarrow\infty$,
	\[
	\left| \dfrac{\sum_{\theta(s)=0} \left[I(p_{\text{weighted}}(s, \widetilde{k}) \leq t) - \Pr(p_{\text{weighted}}(s, \widetilde{k}) \leq t \mid \theta(s)=0) \right]}{\sum_{s\in\mathbb{S}}\Pr(p_{\text{weighted}}(s, \widetilde{k}) \leq t, \theta(s)=0)} \right| \rightarrow 0, \eqno{(\mathrm{A}.15)}
	\]
	in probability. By using triangular inequality, we can obtain
	\begin{eqnarray*}
		& & \left| \dfrac{\sum_{s\in\mathbb{S}} \left[I(p_{\text{weighted}}(s, \widetilde{k}) \leq t, \theta(s)=0) - \Pr(p_{\text{weighted}}(s, \widetilde{k}) \leq t, \theta(s)=0)\right]}{\sum_{s\in\mathbb{S}}\Pr(p_{\text{weighted}}(s, \widetilde{k}) \leq t, \theta(s)=0)} \right| \\
		&\leq& \left| \dfrac{\sum_{\theta(s)=0} \left[I(p_{\text{weighted}}(s, \widetilde{k}) \leq t) - \Pr(p_{\text{weighted}}(s, \widetilde{k}) \leq t \mid \theta(s)=0) \right]}{\sum_{s\in\mathbb{S}}\Pr(p_{\text{weighted}}(s, \widetilde{k}) \leq t, \theta(s)=0)} \right|\\
		&&+ \left|\dfrac{\sum_{\theta(s)=0}\Pr(p_{\text{weighted}}(s, \widetilde{k}) \leq t \mid \theta(s)=0) - \sum_{s\in\mathbb{S}}\Pr(p_{\text{weighted}}(s, \widetilde{k}) \leq t, \theta(s)=0)}{\sum_{s\in\mathbb{S}}\Pr(p_{\text{weighted}}(s, \widetilde{k}) \leq t, \theta(s)=0)} \right|.
	\end{eqnarray*}
	Combining (A.12) and (A.13), we can conclude the result of (A.14). Note that
	\begin{eqnarray*}
		\sum_{s\in\mathbb{S}}\Pr(p_{\text{weighted}}(s, \widetilde{k}) \leq t, \theta(s)=0) &=& \sum_{s\in\mathbb{S}}\Pr(p_{\text{weighted}}(s, \widetilde{k}) \leq t \mid \theta(s)=0)\Pr(\theta(s)=0) \\
		&\leq& \sum_{s\in\mathbb{S}} \left(\frac{\pi_1(s)}{1-\pi_1(s)}\right)^{1/\widetilde{k}}(1-\pi_1(s))t\\
		&=& \sum_{s\in\mathbb{S}} (1-\pi_1(s))^{1-1/\widetilde{k}} \pi_1(s)^{1/\widetilde{k}} t.
	\end{eqnarray*}
	This together with (A.14) prove the result of Theorem 1.
\end{proof}

\subsection*{A.2 Proof of Theorem 2}

\begin{proof}
	Recall the definitions of $\mathrm{mFDR} (\boldsymbol{\delta}_{\text{rescaled}}^{\widetilde{k}}(t))$ and $\mathrm{mFDR} (\boldsymbol{\delta}(t))$. It is easy to see that, to prove $\mathrm{mFDR} (\boldsymbol{\delta}_{\text{rescaled}}^{\widetilde{k}}(t^{1}_{\text{OR}})) \leq \mathrm{mFDR} (\boldsymbol{\delta}(t^{1}_{\text{OR}}))$, it suffices to show that
	\[
	\mathrm{ETP} (\boldsymbol{\delta}_{\text{rescaled}}^{\widetilde{k}}(t))\geq \mathrm{ETP} (\boldsymbol{\delta}(t)),\eqno{(\mathrm{A}.16)}
	\]
	that is, $\sum_{i=1}^m \pi_1(s_i) F_{1, i}(\varphi_{\widetilde{k}}(\pi_1(s_i))\omega_{\widetilde{k}}t)\geq \sum_{i=1}^m \pi_1(s_i) F_{1, i} (t)$. By Assumptions (A3) and (A4), we have that
	{\small
		\begin{eqnarray*}
			& & \sum_{i=1}^m \pi_1(s_i) F_{1, i}(\varphi_{\widetilde{k}}(\pi_1(s_i))\omega_{\widetilde{k}}t)\\
			&=& \sum_{i=1}^m \pi_{1}(s_i) \sum_{i=1}^m \frac{\pi_{1}(s_i)}{\sum_{i=1}^m \pi_{1}(s_i)} F_{1, i}\left(\dfrac{t}{\varphi_{\widetilde{k}}^{-1}(\pi_1(s_i))\omega_{\widetilde{k}}^{-1}}\right)\\
			&\geq& \sum_{i=1}^m \pi_{1}(s_i) F_{1, i} \left(\dfrac{t}{\sum_{i=1}^m \frac{\pi_{1}(s_i)}{\sum_{i=1}^m \pi_{1}(s_i)}\varphi_{\widetilde{k}}^{-1}(\pi_1(s_i))\omega_{\widetilde{k}}^{-1}}\right)\\
			&=& \sum_{i=1}^m \pi_{1}(s_i) F_{1, i} \left( \dfrac{\sum_{i=1}^m (1-\pi_{1}(s_i))\sum_{i=1}^m \pi_{1}(s_i)t}{\sum_{i=1}^m \left[(1-\pi_{1}(s_i))^{1-1/\widetilde{k}}\pi_{1}(s_i)^{1/\widetilde{k}}\right]\sum_{i=1}^m \left[(1-\pi_{1}(s_i))^{1/\widetilde{k}}\pi_{1}(s_i)^{1-1/\widetilde{k}}\right]} \right)\\
			&\geq& \sum_{i=1}^m \pi_1(s_i) F_{1, i} (t).
		\end{eqnarray*}
	}
	Hence (A.16) is proved. Combining (A.16) and the definition of $t^{1}_{\text{OR}}$, we have that
	\[
	\mathrm{mFDR} (\boldsymbol{\delta}_{\text{rescaled}}^{\widetilde{k}}(t^{1}_{\text{OR}})) \leq \mathrm{mFDR} (\boldsymbol{\delta}(t^{1}_{\text{OR}})) \leq \alpha.
	\]
	By the definition of $t^{\widetilde{k}}_{\text{OR}}$, we have that $t^{\widetilde{k}}_{\text{OR}}\geq t^{1}_{\text{OR}}$. This together with (A.16) yield that
	\[
	\mathrm{ETP} (\boldsymbol{\delta}_{\text{rescaled}}^{\widetilde{k}}(t^{\widetilde{k}}_{\text{OR}})) \geq \mathrm{ETP} (\boldsymbol{\delta}_{\text{rescaled}}(t^{1}_{\text{OR}})) \geq \mathrm{ETP} (\boldsymbol{\delta}(t^{1}_{\text{OR}})).
	\]
\end{proof}

\subsection*{A.3 Estimation of Local Sparsity Levels}

We conduct simulations to evaluate the estimation of the local sparsity levels via the kernel-smooth Lfdr statistic. Consider the one-dimensional settings described in Section 3.1 with $\pi_{1d}$ fixed at $0.6$ and $\mu$ fixed at $2$. The estimation of the local sparsity levels are displayed in Figure 5 and the kernel-smooth method is denoted by `STRAW'. We can see that both estimation methods of LAWS and STRAW are somewhat conservative, while STRAW better estimates the spatial structure with a flat peak in the estimate curve of $\pi_{1d}$.

\begin{figure}[htp]
	\centering
	\includegraphics[height=1.8in,width=6in]{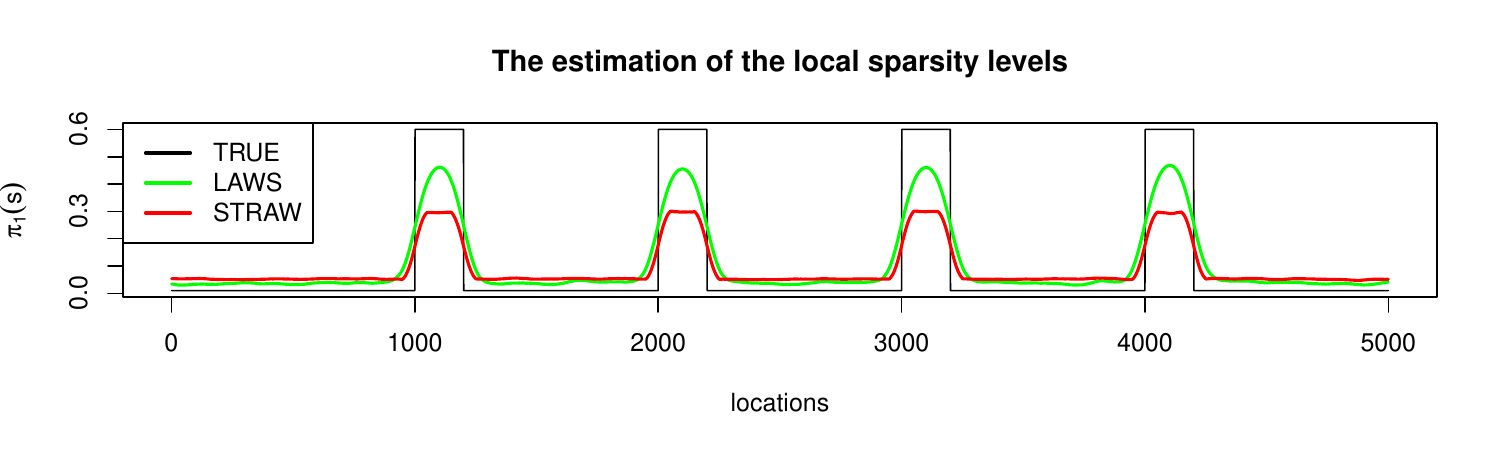}\\[-5mm]
	\caption{\footnotesize{True $\pi_1(s)$ v.s. estimated $\pi_1(s)$. }}
	\label{fig:4}
\end{figure}

\subsection*{A.4 Simulation Studies in Three-Dimensional Settings}

\par
Consider simultaneous testing of $10000=20\times20\times25$ null hypotheses located at a cubic region of three-dimensional Euclidean space. The local sparsity levels $\{\pi_1(s)\}_{20\times20\times25}$ are set to
\begin{eqnarray*}
	\pi_1(s) &=& \pi_{3d}, \text{~for~} s\in\{6, 7, \cdots, 15\}\times\{11, 12, \cdots, 20\}\times\{11, 12, \cdots, 20\};\\
	\pi_1(s) &=& 0.01, \text{~for~other~locations},
\end{eqnarray*}
Let $\{\theta(s)\}_{20\times20\times25}$ and $\{X(s)\}_{20\times20\times25}$ be generated from
\begin{eqnarray*}
	\theta(s)          &\sim& \text{Bernoulli}(\pi_1(s)), \text{~for~} s\in\{1,\cdots,20\}\times\{1,\cdots,20\}\times\{1,\cdots,25\},\\
	X(s)\mid\theta(s)  &\sim& (1-\theta(s))N(0, 1) + \theta(s)N(\mu, 1),\\
	&&~~~~~~~~~~~~~~~~\text{~for~} s\in\{1,\cdots,20\}\times\{1,\cdots,20\}\times\{1,\cdots,25\}.
\end{eqnarray*}
Consider the following two scenarios.

\par
{\bf Scenario 5:} fix $\pi_{3d}=0.6$ and change $\mu$ from $1.5$ to $2.0$.

\par
{\bf Scenario 6:} fix $\mu=2.0$ and change $\pi_{3d}$ from $0.4$ to $0.6$.

\par
The corresponding simulation results are displayed in Figure 6. Since the simulation results are consistent with those in one-dimensional and two-dimensional settings, we will not go into further detail.

\begin{figure}[htp]
	\centering
	\begin{subfigure}{0.48\textwidth}
		\includegraphics[width=2.2in]{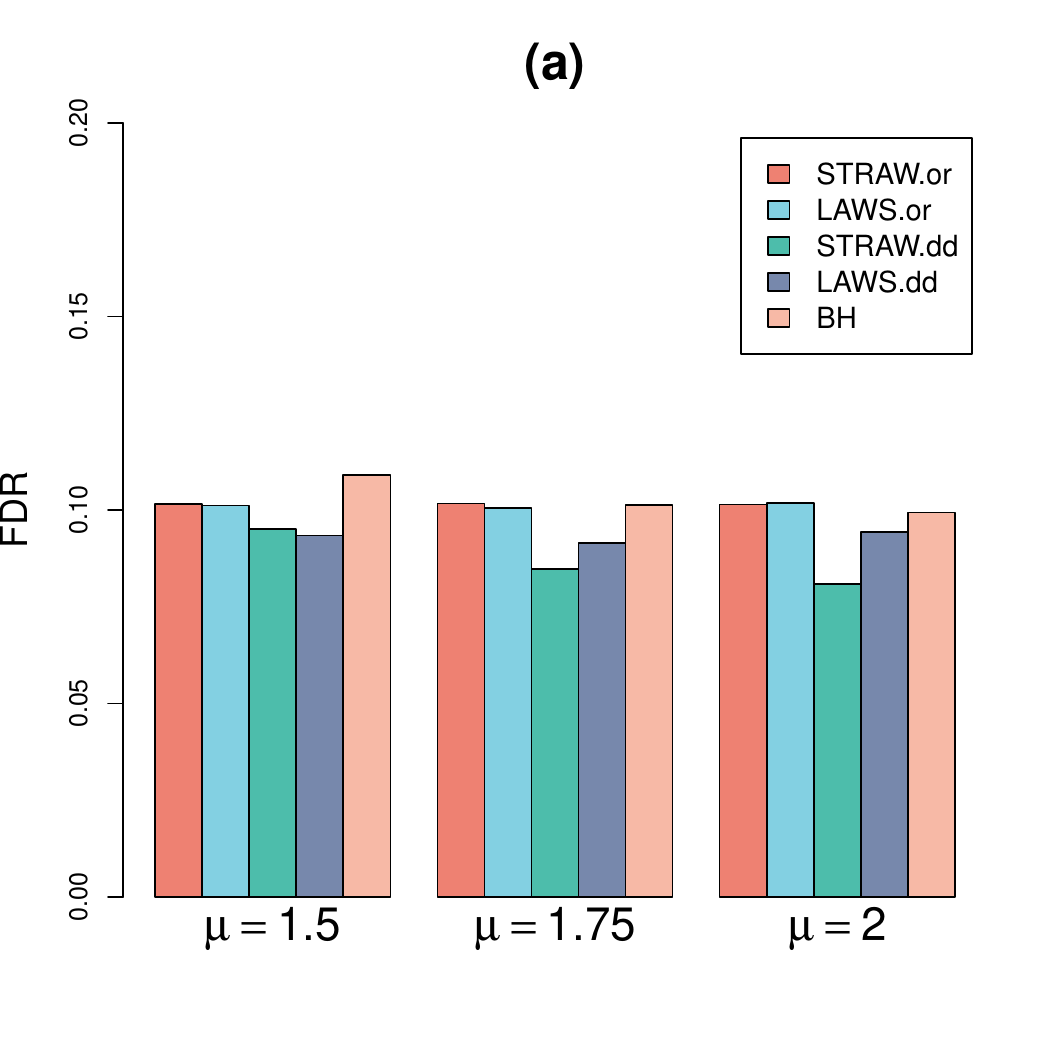}
	\end{subfigure}
	\begin{subfigure}{0.48\textwidth}
		\includegraphics[width=2.2in]{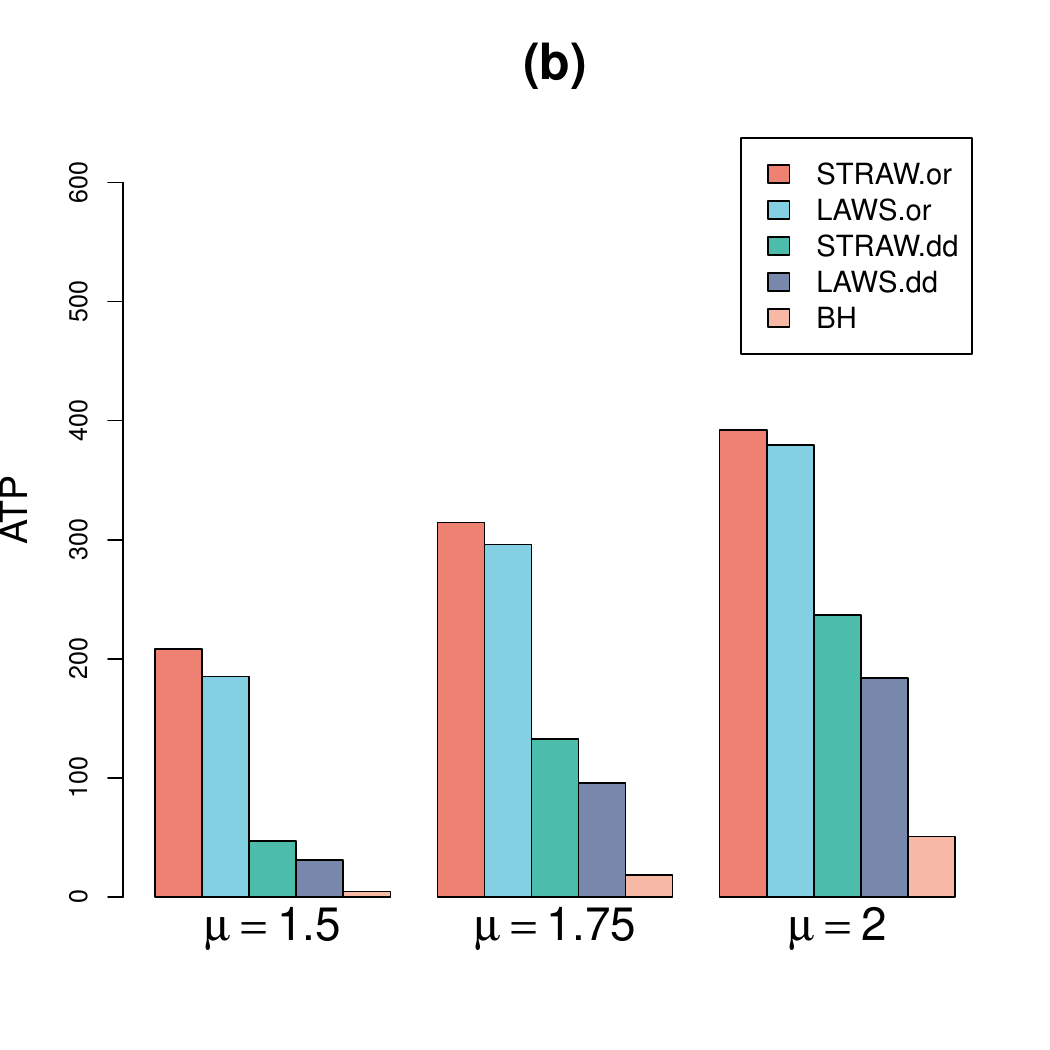}
	\end{subfigure}\\
\begin{subfigure}{0.48\textwidth}
	\includegraphics[width=2.2in]{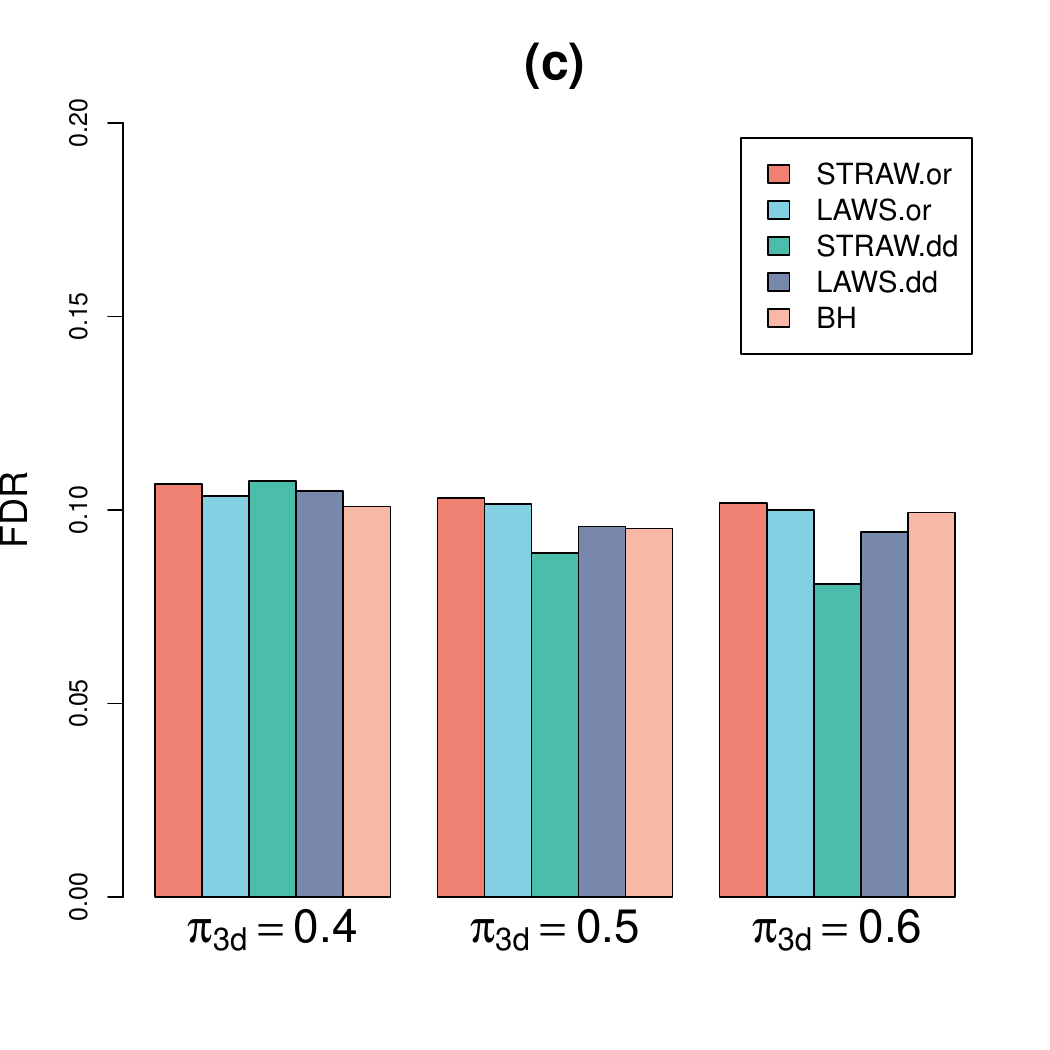}
\end{subfigure}
	\begin{subfigure}{0.48\textwidth}
		\includegraphics[width=2.2in]{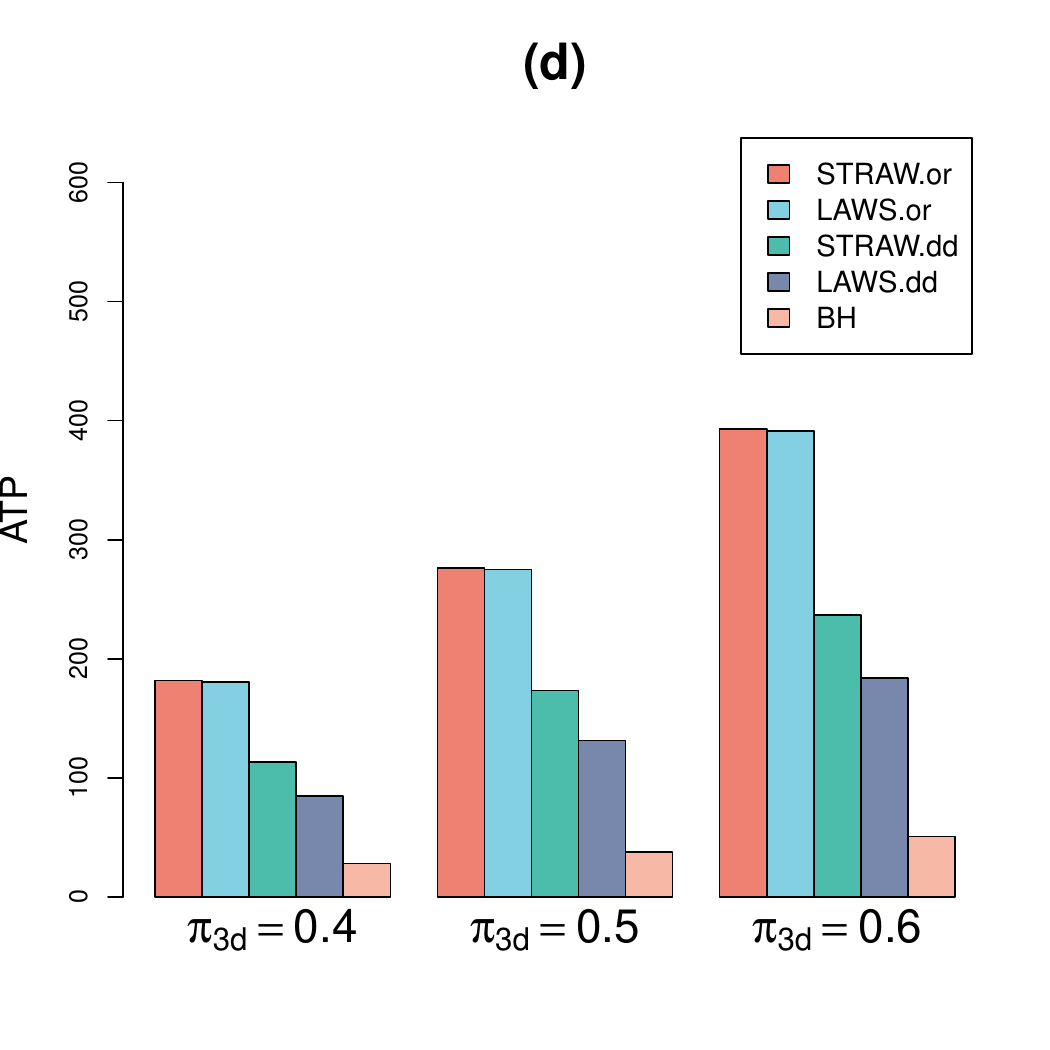}
	\end{subfigure}
	\centering
	\caption{ \footnotesize Simulation results in three-dimensional setting: (a)-(b) simulation results in Scenario 5; (c)-(d) simulation results in Scenario 6.}
	\label{fig:6}
\end{figure}

\newpage
\section*{References}

\begin{description}
	\item Benjamini, Y., \& Hochberg, Y. (1995). Controlling the false discovery rate: a practical and powerful approach to multiple testing. {\sl Journal of the Royal Statistical Society: Series B: Statistical Methodological}, {\bf 57}(1), 289--300.
	\item Benjamini, Y. \& Hochberg, Y. (2000). On the adaptive control of the false discovery rate in multiple testing with independent statistics. {\sl Journal of Educational and Behavioral Statistics}, {\bf 25}(1), 60--83.
	\item Cai, T. T., Sun, W., \& Xia, Y. (2022). LAWS: A locally adaptive weighting and screening approach to spatial multiple testing. {\sl Journal of the American Statistical Association}, {\bf 117}(539), 1370--1383.
	\item Cao, H., Chen, J., \& Zhang, X. (2022). Optimal false discovery rate control for large scale multiple testing with auxiliary information. {\sl The Annals of Statistics}, {\bf 50}(2), 807--857.
	\item Cui, T., Wang, P., \& Zhu, W. (2021). Covariate-adjusted multiple testing in genome-wide association studies via factorial hidden Markov models. {\sl Test}, {\bf 30}(3), 737--757.
	\item Deng, L., He, K., \& Zhang, X. (2022). Powerful spatial multiple testing via borrowing neighboring information. {\sl arXiv preprint arXiv:2210.17121}.
	\item Efron, B., Tibshirani, R., Storey, J. D., \& Tusher, V. (2001). Empirical Bayes analysis of a microarray experiment. {\sl Journal of the American Statistical Association}, {\bf 96}(456), 1151--1160.
	\item G$\ddot{\text{o}}$lz, M., Zoubir, A. M., \& Koivunen, V. (2022). Improving inference for spatial signals by contextual false discovery rates. {\sl In ICASSP 2022-2022 IEEE International Conference on Acoustics, Speech and Signal Processing (ICASSP)}, 5967--5971.
	\item Hu, J. X., Zhao, H., \& Zhou, H. H. (2010). False discovery rate control with groups. {\sl Journal of the American Statistical Association}, {\bf 105}(491), 1215--1227.
	\item Ignatiadis, N., Klaus, B., Zaugg, J. B., \& Huber, W. (2016). Data-driven hypothesis weighting increases detection power in genome-scale multiple testing. {\sl Nature Methods}, {\bf 13}(7), 577--580.
	\item Kuan, P. F. \& Chiang, D. Y. (2012). Integrating prior knowledge in multiple testing under dependence with applications to detecting differential DNA methylation. {\sl Biometrics}, {\bf 68}(3), 774--783.
	\item Li, A. \& Barber, R. F. (2019). Multiple testing with the structure-adaptive Benjamini?Hochberg algorithm. {\sl Journal of the Royal Statistical Society Series B: Statistical Methodology}, {\bf 81}(1), 45--74.
	\item Liang, Z., Cai, T. T., Sun, W., \& Xia, Y. (2022). Locally adaptive transfer learning algorithms for large-scale multiple testing. {\sl arXiv preprint arXiv:2203.11461}.
	\item Sesia, M., Bates, S., Cand{\`e}s, E., Marchini, J., and Sabatti, C. (2021). False discovery rate control in genome-wide association studies with population structure. {\sl Proceedings of the National Academy of Sciences}, {\bf 118}(40), e2105841118.
	\item Stein, M. L. (1999). Interpolation of spatial data: some theory for kriging. {\sl Springer Science \& Business Media}.
	\item Storey, J. D. (2002). A direct approach to false discovery rates. {\sl Journal of the Royal Statistical Society Series B: Statistical Methodology}, {\bf 64}(3), 479--498.
	\item Sun, W. \& Cai, T. T. (2007). Oracle and adaptive compound decision rules for false discovery rate control. {\sl Journal of the American Statistical Association}, {\bf 102}(479), 901--912.
	\item Sun, W. \& Cai, T. T. (2009). Large-scale multiple testing under dependence. {\sl Journal of the Royal Statistical Society Series B: Statistical Methodology}, {\bf 71}(2), 393--424.
	\item Sun, W., Reich, B. J., Cai, T. T., Guindani, M., \& Schwartzman, A. (2015). False discovery control in large-scale spatial multiple testing. {\sl Journal of the Royal Statistical Society Series B: Statistical Methodology}, {\bf 77}(1), 59?-83.
	\item Tansey, W., Koyejo, O., Poldrack, R. A., \& Scott, J. G. (2018). False discovery rate smoothing. {\sl Journal of the American Statistical Association}, {\bf 113}(523), 1156--1171.
	\item Wang, P. \& Zhu, W. (2019). Replicability analysis in genome-wide association studies via Cartesian hidden Markov models. {\sl BMC Bioinformatics}, {\bf 20}(1), 146.
	\item Wang, X., Shojaie, A., \& Zou, J. (2019). Bayesian hidden Markov models for dependent large-scale multiple testing. {\sl Computational Statistics and Data Analysis}, {\bf 136}, 123--136.
	\item Wei, Z., Sun, W., Wang, K., \& Hakonarson, H. (2009). Multiple testing in genome-wide association studies via hidden Markov models. {\sl Bioinformatics}, {\bf 25}(21), 2802--2808.
	\item Xia, Y., Cai, T. T., \& Sun, W. (2020). GAP: a general framework for information pooling in two-sample sparse inference. {\sl Journal of the American Statistical Association}, {\bf 115}, 1236?1250.
	\item Yun, S., Zhang, X., \& Li, B. (2022). Detection of local differences in spatial characteristics between two spatiotemporal random fields. {\sl Journal of the American Statistical Association}, {\bf 117}(537), 291--306.
	\item Zang, Y. F., He, Y., Zhu, C. Z., Cao, Q. J., Sui, M. Q., Liang, M., Tian, L. X., Jiang, T. Z., \& Wang, Y.-F. (2007). Altered baseline brain activity in children with ADHD revealed by resting-state functional MRI. {\sl Brain and Development}, {\bf 29}(2), 83--91.
	\item Zhang, X. \& Chen, J. (2022). Covariate adaptive false discovery rate control with applications to omics-wide multiple testing. {\sl Journal of the American Statistical Association}, {\bf 117}(537), 411--427.
	\item Zou, Q. H., Zhu, C. Z., Yang, Y. H., Zuo, X. N., Long, X. Y., Cao, Q. J., Wang, Y. F., \& Zang, Y. F. (2008). An improved approach to detection of amplitude of low-frequency fluctuation (ALFF) for resting-state fMRI: fractional ALFF. {\sl Journal of Neuroscience Methods}, {\bf 172}(1), 137--141.
\end{description}

\end{document}